\documentclass[a4paper,UKenglish,cleveref, autoref, thm-restate, pdfa]{lipics-v2021}
\title{An Internal Language for Categories Enriched over Generalised
  Metric Spaces} %TODO Please add
\titlerunning{An Internal Language for Categories Enriched over Generalised
  Metric Spaces} %TODO optional, please use if title is longer than one line

\author{Fredrik Dahlqvist}{University College London, United Kingdom}{f.dahlqvist@ucl.ac.uk}{}{}%TODO mandatory, please use full name; only 1 author per \author macro; first two parameters are mandatory, other parameters can be empty. Please provide at least the name of the affiliation and the country. The full address is optional

\author{Renato Neves}{University of Minho \& INESC-TEC, Portugal}{nevrenato@gmail.com}{}{}
\authorrunning{F. Dahlqvist and R. Neves} %TODO mandatory. First: Use abbreviated first/middle names. Second (only in severe cases): Use first author plus 'et al.'

\Copyright{Fredrik Dahlqvist and Renato Neves} %TODO mandatory, please use full first names. LIPIcs license is "CC-BY";  http://creativecommons.org/licenses/by/3.0/

 %TODO mandatory; please add comma-separated list of keywords

\category{} %optional, e.g. invited paper

\relatedversion{} %optional, e.g. full version hosted on arXiv, HAL, or other respository/website
\acknowledgements{This work was financed by the ERDF – European
  Regional Development Fund through the Operational Programme for
  Competitiveness and Internationalisation - COMPETE 2020 Programme
  and by National Funds through the Portuguese funding agency, FCT -
  Fundação para a Ciência e a Tecnologia, within project
  POCI-01-0145-FEDER-030947. This work was also financed by the UK
  Research Institute in Verified Trustworthy Software Systems, via
  project ``Quantitative Algebraic Reasoning for Hybrid Programs'',
  and by the Leverhulme Project Grant ``Verification of Machine
  Learning Algorithms''. The authors are very grateful for the reviewer's
  incisive feedback.}%optional

\nolinenumbers %uncomment to disable line numbering

\hideLIPIcs  %uncomment to remove references to LIPIcs series (logo, DOI, ...), e.g. when preparing a pre-final version to be uploaded to arXiv or another public repository

\EventEditors{John Q. Open and Joan R. Access}
\EventNoEds{2}
\EventLongTitle{42nd Conference on Very Important Topics (CVIT 2016)}
\EventShortTitle{CVIT 2016}
\EventAcronym{CVIT}
\EventYear{2016}
\EventDate{December 24--27, 2016}
\EventLocation{Little Whinging, United Kingdom}
\EventLogo{}
\SeriesVolume{42}
\ArticleNo{23}
\usepackage{proof}
\usepackage{tikz-cd}
\usepackage{stmaryrd}
\usepackage{bbm}
\usepackage{mathtools}
\usepackage{xspace}
\DeclarePairedDelimiter\abs{\lvert}{\rvert}%
\DeclarePairedDelimiter\norm{\lVert}{\rVert}%
\makeatletter
\let\oldabs\abs
\def\abs{\@ifstar{\oldabs}{\oldabs*}}
\let\oldnorm\norm
\def\norm{\@ifstar{\oldnorm}{\oldnorm*}}
\makeatother

\newcommand{\defeq}{\triangleq}
\newcommand{\Cats}{\catfont{Cat}}

\newcommand{\Met}{\catfont{Met}}

\newcommand{\Free}[1]{\sfunfont{F}}

\newcommand{\typeI}{\typefont{I}}

\newcommand{\Shuff}{\mathrm{Sf}}
\newcommand{\VCat}{\mathcal{V}\text{-}\Cats}
\newcommand{\VCatSy}{\mathcal{V}\text{-}\Cats_{\mathsf{sym}}}
\newcommand{\VCatSe}{\mathcal{V}\text{-}\Cats_{\mathsf{sep}}}
\newcommand{\VCatSS}{\mathcal{V}\text{-}\Cats_{\mathsf{sym,sep}}}
\newcommand{\Lang}{\mathrm{Lang}}
\newcommand{\Syn}{\mathrm{Syn}}
\newcommand{\catfont}[1]{\mathsf{#1}}

\newcommand{\catC}{\catfont{C}}
\newcommand{\catD}{\catfont{D}}

\newcommand{\Set}{\catfont{Set}}

\newcommand{\Pos}{\catfont{Pos}}

\newcommand{\Subs}[2]{\catfont{Sub}_{}}

\newcommand{\Aut}{\catfont{Aut}}
\newcommand{\funfont}[1]{#1}
\newcommand{\funF}{\funfont{F}}

\newcommand{\funG}{\funfont{G}}

\newcommand{\sfunfont}[1]{\mathrm{#1}}

\newcommand\adjunct[2]{\xymatrix@=8ex{\ar@{}[r]|{\top}\ar@<1mm>@/^2mm/[r]^{{#2}}
& \ar@<1mm>@/^2mm/[l]^{{#1}}}}
\newcommand\adjunctop[2]{\xymatrix@=8ex{\ar@{}[r]|{\bot}\ar@<1mm>@/^2mm/[r]^{{#2}}
& \ar@<1mm>@/^2mm/[l]^{{#1}}}}
\newcommand\retract[2]{\xymatrix@=8ex{\ar@{}[r]|{}\ar@<1mm>@/^2mm/@{^{(}->}[r]^{{#2}}
& \ar@<1mm>@/^2mm/@{->>}[l]^{{#1}}}}
\newcommand{\sem}[1]{\llbracket #1 \rrbracket}

\newcommand{\app}{\mathrm{app}}

\newcommand{\comp}{\cdot}

\DeclareMathOperator{\id}{\mathsf{id}}
\DeclareMathOperator{\sw}{\mathsf{sw}}
\DeclareMathOperator{\spl}{\mathsf{sp}}
\DeclareMathOperator{\sh}{\mathsf{sh}}
\DeclareMathOperator{\join}{\mathsf{jn}}
\DeclareMathOperator{\exch}{\mathsf{exch}}
\newcommand{\resp}{resp.\xspace~}

\newcommand{\ie}{\textit{i.e.}\xspace~}
\newcommand{\eg}{\textit{e.g.}\xspace~}
\newcommand{\Nats}{\mathbb{N}}

\newcommand{\typefont}[1]{\mathbb{#1}}

\newcommand{\typeA}{\typefont{A}}
\newcommand{\typeB}{\typefont{B}}
\newcommand{\typeC}{\typefont{C}}

\newcommand{\typeD}{\typefont{D}}
\newcommand{\rulename}[1]{(\mathrm{\mathbf{#1}})}
\newcommand{\vljud}{\rhd}

\newcommand{\closValBP}[2]{\mathsf{Values}(#1, #2)}
\newcommand{\prog}[1]{\mathtt{#1}}
\newcommand{\real}{\mathtt{Real}}
\newcommand{\bern}{\mathtt{bernoulli}}
\newcommand{\normal}{\mathtt{normal}}

\newcommand{\Ban}{\mathsf{Ban}}
\newcommand{\ptp}{\widehat{\otimes}_{\pi}}
\newcommand{\meas}{\mathcal{M}}
\newcommand{\R}{\mathbb{R}}
\newcommand{\unit}{\mathtt{unit}}

\allowdisplaybreaks[3]
\begin{document}

\maketitle

\begin{abstract}
  Programs with a \emph{continuous} state space or that interact with
  physical processes often require notions of equivalence going beyond
  the standard binary setting in which equivalence either holds or
  does not hold. In this paper we explore the idea of equivalence
  taking values in a quantale $\mathcal{V}$, which covers the cases of
  (in)equations and (ultra)metric equations among others.

  Our main result is the introduction of a
  \emph{$\mathcal{V}$-equational deductive system} for linear
  $\lambda$-calculus together with a proof that it is sound and
  complete (in fact, an internal language) for a class of enriched
  autonomous categories. In the case of inequations, we get an
  internal language for autonomous categories enriched over partial
  orders. In the case of (ultra)metric equations, we get an internal
  language for autonomous categories enriched over (ultra)metric
  spaces.

  We use our results to obtain examples of inequational and metric
  equational systems for higher-order programs that contain real-time
  and probabilistic behaviour.
\end{abstract}

\section{Introduction}

%\subsection{Motivation  and contributions}

Programs frequently act over a \emph{continuous} state space or
interact with physical processes like time progression or the movement
of a vehicle. Such features naturally call for notions of
approximation and refinement integrated in different aspects of
program equivalence. Our paper falls in this line of research.
Specifically, our aim is to integrate notions of approximation and
refinement into the \emph{equational system} of linear
$\lambda$-calculus~\cite{benton92,mackie93,maietti05}.

The core idea that we explore in this paper is to have equations
$t =_q s$ labelled by elements $q$ of a quantale $\mathcal{V}$. This
covers a wide range of situations, among which the cases of
(in)equations~\cite{kurz2017quasivarieties,adamek20} and metric
equations~\cite{mardare16,mardare17}. The latter case is perhaps less
known: it consists of equations $t =_{\epsilon}s$ labelled by a
non-negative rational number $\epsilon$ which represents the `maximum
distance' that the two terms $t$ and $s$ can be from each other. In
order to illustrate metric equations, consider a programming language
with a (ground) type $X$ and a signature of operations
$\Sigma = \{ \prog{wait_n} : X \to X \mid n \in \Nats \}$ that model
time progression over computations of type $X$. Specifically,
$\prog{wait_n}(x)$ reads as ``add a latency of $n$ seconds to the
computation $x$''. In this context, the following axioms involving
metric equations arise naturally:
\begin{flalign}\label{ax}
  \prog{wait_0}(x) =_0 x \hspace{1cm} \prog{wait_n}(\prog{wait_m}(x)) =_0
  \prog{wait_{n + m }}(x) \hspace{1cm}
  \infer{\prog{wait_n}(x) =_\epsilon \prog{wait_m}(x)}{\epsilon = |m - n|}
\end{flalign}
An equation $t =_0 s$ states that the terms $t$ and $s$ are exactly the
same and equations $t =_\epsilon s$ state that $t$ and $s$ differ by
\emph{at most} $\epsilon$ seconds in their execution time.

\noindent
\textbf{Contributions.}  In this paper we introduce an equational
deductive system for linear $\lambda$-calculus in which equations are
labelled by elements of a quantale $\mathcal{V}$.  By using key
features of a quantale's structure, we show that this deductive system
is \emph{sound and complete} for a class of enriched symmetric
monoidal closed categories (\ie enriched autonomous categories).  In
particular, if we fix $\mathcal{V}$ to be the Boolean quantale this
class of categories consists of autonomous categories enriched over
partial orders. If we fix $\mathcal{V}$ to be the (ultra)metric
quantale, this class of categories consists of autonomous categories
enriched over (ultra)metric spaces. The aforementioned example of wait
calls fits in the setting in which $\mathcal{V}$ is the metric
quantale.  Our result provides this example with a sound and complete
metric equational system, where the models are all those autonomous
categories enriched over metric spaces that can soundly interpret the
axioms of wait calls~\eqref{ax}.

The next contribution of our paper falls in one of the major topics of
categorical logic: to establish \emph{logical descriptions} of certain
classes of categories. A famous result of this kind is the
correspondence between $\lambda$-calculus and Cartesian closed
categories which states that the former is the \emph{internal
  language} of the latter~\cite{lambek88} -- such a correspondence
allows to study Cartesian closed categories by means of logical
tools. An analogous result is presented in~\cite{mackie93,maietti05}
for linear $\lambda$-calculus and symmetric monoidal closed (\ie
autonomous) categories. We show that linear $\lambda$-calculus
equipped with a $\mathcal{V}$-equational system is the internal
language of autonomous categories enriched over `generalised metric
spaces'.

\noindent
\textbf{Outline.}  Section~\ref{sec:internal} recalls linear
$\lambda$-calculus, its equational system, and the famous
correspondence to autonomous categories, via soundness, completeness,
and internal language theorems.  The contents of this section are
slight adaptations of results presented in~\cite{mackie93,benton92},
the main difference being that we forbid the exchange rule to be
explicitly part of linear $\lambda$-calculus (instead it is only
admissible). This choice is important to ensure that judgements in the
calculus have \emph{unique} derivations, which allows us to refer to
their interpretations unambiguously~\cite{shulman19}.
Section~\ref{sec:main} presents the main contributions of this
paper. It walks a path analogous to Section~\ref{sec:internal}, but
now in the setting of $\mathcal{V}$-equations (\ie equations labelled
by elements of a quantale $\mathcal{V}$). As we will see, the semantic
counterpart of moving from equations to $\mathcal{V}$-equations is to
move from categories to categories enriched over
\emph{$\mathcal{V}$-categories}. The latter, often regarded as
generalised metric spaces, are central entities in a fruitful area of
enriched category theory that aims to treat uniformly different kinds
of `structured sets', such as partial orders, fuzzy partial orders,
and (ultra)metric spaces~\cite{lawvere73,stubbe14,velebil19}. Our
results are applicable to all these
examples. Section~\ref{sec:examples} presents some examples of
$\mathcal{V}$-equational axioms and corresponding
models. Specifically, we will revisit the axioms of wait
calls~\eqref{ax} and consider an inequational variant. Then we will
study a metric axiom for probabilistic programs and show that the
category of Banach spaces and short linear maps is a model for the
resulting metric theory. We will additionally use this example to
illustrate how our deductive system allows to compute an
approximate distance between two probabilistic programs easily as opposed to
computing an exact distance `semantically' which tends to involve
quite complex operators.  Finally, Section~\ref{sec:concl} establishes
a functorial connection between our results and previous well-known
semantics for linear logic~\cite{paiva99, maietti05}, and concludes
with a brief exposition of future work.  We assume knowledge of
$\lambda$-calculus and category
theory~\cite{mackie93,maietti05,lambek88,maclane98}. Proofs omitted
in the main text are available in the appendix.

% Due to
% space restrictions some proofs were moved to the appendix.

\noindent
\textbf{Related work.}  Several approaches to incorporating
quantitative information to programming languages have been explored
in the literature. Closest to this work are various approaches
targeted at $\lambda$-calculi.
In~\cite{crubille2015metric,crubille2017metric} a notion of distance
called \emph{context distance} is developed, first for an affine, then
for a more general $\lambda$-calculus, with probabilistic programs as
the main motivation. \cite{gavazzo18} considers a notion of
quantale-valued \emph{applicative (bi)similarity}, an operational
\emph{coinductive} technique used for showing contextual equivalence
between two programs.  Recently, \cite{pistone21} presented several
Cartesian closed categories of generalised metric spaces that provide
a quantitative semantics to simply-typed $\lambda$-calculus based on a
generalisation of logical relations.  None of these examples reason
about distances in a quantitative equational system, and in this
respect our work is closer to the metric universal algebra developed
in \cite{mardare16,mardare17}.

A different approach consists in encoding quantitative information via
a type system.  In particular, graded (modal) types
\cite{girard1992bounded,gaboardi2016combining,orchard2019quantitative}
have found applications in \eg differential privacy~\cite{reed10} and
information flow \cite{abadi1999core}.  This approach is to some
extent orthogonal to ours as it mainly aims to model coeffects, whilst
we aim to reason about the intrinsic quantitative nature of
$\lambda$-terms acting \eg on continuous or ordered spaces.

Quantum programs provide an interesting example of intrinsically
quantitative programs, by which we mean that the metric structure on
quantum states does not arise from (co)effects.  Recently,
\cite{hung19} showed how the issue of noise in a quantum
while-language can be handled by developing a \emph{deductive system}
to determine how similar a quantum program is from its idealised,
noise-free version; an approach very much in the spirit of this work.

\section{An internal language for autonomous categories}
\label{sec:internal}
In this section we briefly recall linear $\lambda$-calculus, which can
be regarded as a term assignment system for the exponential free,
multiplicative fragment of \emph{intuitionistic linear logic}. Then we
recall that it is sound and complete w.r.t. autonomous categories, and
also that it is an internal language for such categories. We
mention only what is needed to present
our results, the interested reader will find a more detailed
exposition in~\cite{mackie93,benton92,maietti05}. Let
us start by fixing a \emph{class} $G$ of ground types.  The grammar of
types for linear $\lambda$-calculus is given by: \vspace{-1.85pt}
\[
  \typeA ::=  X \in G \mid \typeI \mid 
  \typeA \otimes \typeA \mid \typeA \multimap \typeA
\] 
We also fix a class $\Sigma$ of sorted operation symbols
$f : \typeA_1,\dots,\typeA_n \to \typeA$ with $n \geq 1$.  As usual,
we use Greek letters $\Gamma,\Delta, E,\dots$ to denote \emph{typing
  contexts}, \ie lists $x_1 : \typeA_1,\dots,x_n : \typeA_n$ of typed
variables such that each variable $x_i$ occurs at most once in
$x_1,\dots,x_n$.

We will use the notion of a shuffle for building a linear typing
system such that the \emph{exchange rule} is admissible and each
judgement $\Gamma \vljud v : \typeA$ (details about these below) has a
\emph{unique derivation} -- this will allow us to refer to a
judgement's denotation $\sem{\Gamma \vljud v : \typeA}$ unambiguously.
By shuffle we mean a permutation of typed variables in a context
sequence $\Gamma_1,\dots,\Gamma_n$ such that for all $i \leq n$ the
relative order of the variables in $\Gamma_i$ is
preserved~\cite{shulman19}. For example, if
$\Gamma_1 = x : \typeA, y : \typeB$ and $\Gamma_2 = z : \typeC$ then
$z : \typeC, x : \typeA, y : \typeB$ is a shuffle but
$y : \typeB, x : \typeA, z : \typeC$ is \emph{not}, because we changed
the order in which $x$ and $y$ appear in $\Gamma_1$.  As explained
in~\cite{shulman19} (and also in the proof of
Lemma~\ref{lem:unique}), such a restriction on relative orders is
crucial for judgements having unique derivations.  We denote by
$\Shuff(\Gamma_1;\dots;\Gamma_n)$ the set of shuffles on
$\Gamma_1,\dots,\Gamma_n$.

The term formation rules of linear $\lambda$-calculus are listed in
Fig.~\ref{fig:lang}. They correspond to the natural deduction rules of
the exponential free, multiplicative fragment of intuitionistic linear
logic. % A somewhat less standard feature in Fig.~\ref{fig:lang} is the
% presence of rule $\rulename{seq}$ which corresponds to sequential
% composition in programming (the term $v\ \prog{to}\ x.\ w$ reads as
% ``bind the computation $v$ to $x$ and then proceed with computation
% $w$''). Its presence is justified by $\rulename{seq}$ being a standard
% feature of programming languages, despite in our case acting merely as
% syntactic sugar for the cut rule (see
% Lemma~\ref{lem:exch_subst}). Note, for instance, that if we have the
% successor function $\prog{succ} : \real \to \real$ in the signature
% $\Sigma$, then the term
% $y : \real \vljud \prog{succ}(y)\ \prog{to}\ x.\ \prog{succ}(x) :
% \real$ increments the value of $y$ by two, in other words it is
% equivalent to the term
% $y : \real \vljud \prog{succ}(\prog{succ}(y)) : \real$ which does not
% explicitly involve sequential composition.

\begin{lemma}\label{lem:unique}
  All judgements $\Gamma \vljud v : \typeA$ have a unique derivation.
\end{lemma}
Substitution is defined in the expected
way,  and the following result is standard.

\begin{lemma}[Exchange and Substitution]\label{lem:exch_subst}
  For every judgement
  $\Gamma, x : \typeA, y : \typeB, \Delta \vljud v : \typeC$ we can derive
  $\Gamma, y : \typeB, x : \typeA, \Delta \vljud v : \typeC$.
  For all judgements
  $\Gamma, x : \typeA \vljud v : \typeB$ and
  $\Delta \vljud w : \typeA$ we can derive
  $\Gamma,\Delta \vljud v[w/x] : \typeB$.
\end{lemma}
We now recall the interpretation of judgements
$\Gamma \vljud v: \typeA$ in a symmetric monoidal closed (autonomous)
category $\catC$. But before proceeding with this description, let us
fix notation for some of the constructions available in autonomous
categories. For all $\catC$-objects $X,Y,Z$,
$\sw : X \otimes Y \to Y \otimes X$ denotes the symmetry morphism,
$\lambda : \typeI \otimes X \to X$ the left unitor,
$\app : (X \multimap Y) \otimes X \to Y$ the application, and
$\alpha : X \otimes (Y \otimes Z) \to (X \otimes Y) \otimes Z$ the
left associator. Moreover for all $\catC$-morphisms
$f : X \otimes Y \to Z$ we denote the corresponding curried version
(right transpose) by $\overline{f} : X \to (Y \multimap Z)$.

For all ground types $X \in G$ we postulate an interpretation
$\sem{X}$ as a $\catC$-object.  Types are then interpreted by
induction over the type structure of linear $\lambda$-calculus, using
the tensor $\otimes$ and exponential $\multimap$ constructs of
autonomous categories. Given a non-empty context
$\Gamma = \Gamma', x : \typeA$, its interpretation is defined by
$\sem{\Gamma', x : \typeA} = \sem{\Gamma'} \otimes \sem{\typeA}$ if
$\Gamma'$ is non-empty and $\sem{\Gamma', x : \typeA} = \sem{\typeA}$
otherwise. The empty context $-$ is interpreted as $\sem{-} = \typeI$
where $\typeI$ is the unit of $\otimes$ in $\catC$.  To keep notation
simple, given $X_1,\dots,X_n \in \catC$ we write
$X_1 \otimes \dots \otimes X_n$ for the $n$-tensor
$(\dots (X_1 \otimes X_2) \otimes \dots ) \otimes X_n$, and similarly
for $\catC$-morphisms.

We will also need some `housekeeping' morphisms to handle 
interactions between context interpretation and the autonomous
structure of $\catC$. Specifically, given contexts
$\Gamma_1,\dots,\Gamma_n$ we denote by
$\spl_{\Gamma_1; \dots;\Gamma_n} : \sem{\Gamma_1, \dots, \Gamma_n} \to
\sem{\Gamma_1} \otimes \dots \otimes \sem{\Gamma_n}$ the morphism that
splits $\sem{\Gamma_1, \dots, \Gamma_n}$ into
$\sem{\Gamma_1} \otimes \dots \otimes \sem{\Gamma_n}$, and by
$\join_{\Gamma_1;\dots;\Gamma_n}$ the corresponding inverse. Given a
context $\Gamma, x : \typeA, y : \typeB, \Delta$ we denote by
$\exch_{\Gamma, \underline{x : \typeA, y : \typeB}, \Delta} :
\sem{\Gamma, x : \typeA, y : \typeB, \Delta} \to \sem{\Gamma, y :
  \typeB, x : \typeA, \Delta}$ the morphism corresponding to the
permutation of the variable $x : \typeA$ with $y : \typeB$.  Whenever
convenient we will drop variable names in the subscripts of $\spl$,
$\join$, and $\exch$. For a context
$E \in \Shuff(\Gamma_1,\dots,\Gamma_n)$ the morphism
$\sh_E : \sem{E} \to \sem{\Gamma_1, \dots, \Gamma_n} $ denotes the
corresponding shuffling morphism.

For every operation symbol $f : \typeA_1, \dots, \typeA_n \to \typeA$
in $\Sigma$ we postulate an interpretation
$\sem{f} : \sem{\typeA_1} \otimes \dots \otimes \sem{\typeA_n} \to
\sem{\typeA}$ as a $\catC$-morphism. The interpretation of judgements
is defined by induction over the structure of judgement
derivation according to the rules in Fig.~\ref{fig:lang_sem}.

\begin{figure*}[h!]
  % \vspace{-10pt}
    \begin{flalign*}
      \infer[\rulename{ax}]{E
        \vljud f(v_1,\dots,v_n) : \typeA}
      {\Gamma_i \vljud v_i : \typeA_i
      \quad f : \typeA_1,\dots,\typeA_n \to \typeA \in \Sigma
      \quad E \in \Shuff(\Gamma_1;\dots;\Gamma_n) }
      \hspace{0.5cm}
      \infer[\rulename{hyp}]{x : \typeA \vljud x : \typeA}{}
    \end{flalign*}
    \vspace{-0.7cm}
    \begin{flalign*}
      \hspace{1.8cm}
      \infer[\rulename{\typeI_i}]{- \vljud \ast : \typeI}{}
      \hspace{2cm}
      \infer[\rulename{\typeI_e}]{E  \vljud v \ \prog{to}\ \ast.\ w
      : \typeA}
      {\Gamma \vljud v : \typeI \quad \Delta 
        \vljud w : \typeA \quad E \in \Shuff(\Gamma;\Delta)}
    \end{flalign*}
    \vspace{-0.7cm}
    \begin{flalign*}
      \hspace{3cm}
      \infer[\rulename{\otimes_i}]{E \vljud v \otimes w : \typeA \otimes
        \typeB}{\Gamma \vljud v : \typeA \quad \Delta \vljud w : \typeB
      \quad E \in \Shuff(\Gamma;\Delta)}
    \end{flalign*}
    \vspace{-0.7cm}
    \begin{flalign*}
      \hspace{2cm}
      \infer[\rulename{\otimes_e}]{E \vljud \prog{pm}\ v\ \prog{to}\
      x \otimes y.\ w : \typeC}
      {\Gamma \vljud v : \typeA \otimes \typeB
      \quad \Delta , x : \typeA, y : \typeB \vljud w : \typeC
      \quad E \in \Shuff(\Gamma;\Delta)}
    \end{flalign*}
    \vspace{-0.7cm}
    \begin{flalign*}
      \infer[\rulename{\multimap_i}]{\Gamma \vljud \lambda x : \typeA . \,
        v : \typeA
      \multimap \typeB}
      {\Gamma, x : \typeA \vljud v : \typeB}
      \hspace{1cm}
      \infer[\rulename{\multimap_e}]{E \vljud v \, w : \typeB}
      {\Gamma \vljud  v : \typeA \multimap \typeB \quad
        \Delta \vljud  w : \typeA \quad E \in \Shuff(\Gamma;\Delta)}      
    \end{flalign*}
  \vspace{-16pt}
  \caption{Term formation rules for linear $\lambda$-calculus.}
  \label{fig:lang}
\end{figure*}
\vspace{-3pt}

\begin{figure*}[h!]
    \begin{flalign*}
      \infer[]{\sem{E \vljud f(v_1,\dots,v_n) : \typeA} = \sem{f} \comp
      (m_1 \otimes \dots \otimes m_n) \comp \spl_{\Gamma_1;\dots;\Gamma_n}
      \comp \sh_E}
      {\sem{\Gamma_i \vljud v_i : \typeA_i} = m_i
      \quad f : \typeA_1,\dots,\typeA_n \to \typeA \in \Sigma
      \quad E \in \Shuff(\Gamma_1 \dots \Gamma_n) }
      \hspace{0.35cm}
      \infer[]{\sem{x : \typeA \vljud x : \typeA} = \id_{\sem{\typeA}}}{}
    \end{flalign*}
    \vspace{-0.7cm}
    \begin{flalign*}
      \hspace{0.6cm}
      \infer[]{\sem{- \vljud \ast : \typeI} = \id_{\sem{\typeI}}}{}      
      \hspace{2cm}
      \infer[]{\sem{E  \vljud v \ \prog{to}\ \ast.\ w : \typeA} =
        n \comp \lambda \comp
      (m \otimes \id) \comp \spl_{\Gamma; \Delta} \comp \sh_E}
      {\sem{\Gamma \vljud v : \typeI} = m \quad \sem{\Delta 
          \vljud w : \typeA} = n \quad E \in \Shuff(\Gamma;\Delta)}
    \end{flalign*}
    \vspace{-0.7cm}
    \begin{flalign*}
      \infer[]{\sem{E \vljud v \otimes w : \typeA \otimes \typeB} =
        (m \otimes n) \comp \spl_{\Gamma;\Delta}
        \comp \sh_E}{\sem{\Gamma \vljud v : \typeA} = m
      \quad \sem{\Delta \vljud w : \typeB} = n
      \quad E \in \Shuff(\Gamma;\Delta)} \hspace{0.1cm}
          \infer[]{\sem{\Gamma \vljud \lambda x : \typeA . \, v : \typeA \multimap \typeB} =
        \overline{ (m \comp \join_{\Gamma; \typeA}) }}
      {\sem{\Gamma, x : \typeA \vljud v : \typeB} = m }
      \hspace{0cm}
    \end{flalign*}
    \vspace{-0.7cm}
    \begin{flalign*} \hspace{1.7cm}
      \infer[]{\sem{E \vljud \prog{pm}\ v\ \prog{to}\ x \otimes y.\
          w : \typeC} =
        n \comp \join_{\Delta; \typeA; \typeB}
        \comp \alpha \comp \sw \comp (m \otimes \id) \comp
        \spl_{\Gamma;\Delta} \comp \sh_E}
      {\sem{\Gamma \vljud v : \typeA \otimes \typeB} = m
      \quad \sem{\Delta , x : \typeA, y : \typeB \vljud w : \typeC} = n
      \quad E \in \Shuff(\Gamma;\Delta)}
    \end{flalign*}
    \vspace{-0.7cm}
    \begin{flalign*} \hspace{2.7cm}
      \infer[]{\sem{E \vljud v \, w : \typeB} = \app \comp (m \otimes n)
        \comp \spl_{\Gamma;\Delta} \comp \sh_E}
      {\sem{\Gamma \vljud  v : \typeA \multimap \typeB} = m \quad
      \sem{\Delta \vljud  w : \typeA} = n \quad E \in \Shuff(\Gamma;\Delta)}     
      \end{flalign*}
  	\vspace{-0.8cm}
  \caption{Judgement interpretation on an autonomous category $\catC$.}
  \label{fig:lang_sem}
\end{figure*}
As detailed in~\cite{benton92,mackie93,maietti05}, linear
$\lambda$-calculus comes equipped with a class of equations
(Fig.~\ref{fig:eqs}), specifically \emph{equations-in-context}
$\Gamma \vljud v = w : \typeA$, that corresponds to the axiomatics of
autonomous categories.  As usual, we omit the context and typing
information of the equations in Fig.~\ref{fig:eqs}, which can be
reconstructed in the usual way. 
\begin{figure*}[h!]
  \begin{subfigure}[large]{0.6\textwidth}
  \scalebox{1.12}{  
  \begin{tabular}{r c l}
    $\prog{pm}\ v \otimes w\ \prog{to}\ x \otimes y.\ u$ &
    $=$ & $u[v/x,w/y]$ \\
    $\prog{pm}\ v\ \prog{to}\ x \otimes y.\
    u[x \otimes y / z]$ &
    $=$ & $u[v/z]$ \\
    $\ast\ \prog{to}\ \ast.\ v$ & $=$ & $v$ \\
    $v\ \prog{to}\ \ast.\ w[\ast/ z]$ & $=$ & $w[v/z]$
  \end{tabular}}
  \vspace{-3pt}
  \caption{Monoidal structure}
  \end{subfigure}
  \vspace{0.2cm}
  \begin{subfigure}{0.3\textwidth}
  \scalebox{1.12}{  
  \begin{tabular}{r c l}
      $(\lambda x : \typeA .\ v)\ w$ &$=$& $v[w/x]$ \\
      $\lambda x : \typeA . (v\ x)$ &$=$& $v$
  \end{tabular}}
  \caption{Higher-order structure}
  \end{subfigure}
  \begin{subfigure}{0.7\textwidth}
  \scalebox{1.12}{
  \begin{tabular}{r c l}
    $u[v\ \prog{to} \ast.\ w/z]$ &$=$& $v\ \prog{to}\ \ast.\
    u[w/z]$ \\
    $u[\prog{pm}\ v\ \prog{to}\ x \otimes y.\ w/z]$
    &$=$& $\prog{pm}\ v\ \prog{to}\ x \otimes y.\ u[w/z]$
  \end{tabular}}
  \caption{Commuting conversions}
  \end{subfigure}
  \caption{Equations corresponding to the axiomatics of autonomous
    categories.}
  \label{fig:eqs}
  \vspace{-7pt}
\end{figure*}

\begin{theorem}
  \label{theo:bsound}
  The equations presented in Fig.~\ref{fig:eqs} are sound w.r.t.
  judgement interpretation. Specifically if
  $\Gamma \vljud v = w : \typeA$ is one of the equations in
  Fig.~\ref{fig:eqs} then
  $\sem{\Gamma \vljud v : \typeA} = \sem{\Gamma \vljud w : \typeA}$.
\end{theorem}

\begin{definition}[Linear $\lambda$-theories]\label{defn:theory}
  Consider a tuple $(G,\Sigma)$ consisting of a class $G$ of ground
  types and a class $\Sigma$ of sorted operation symbols.  A linear
  $\lambda$-theory $((G,\Sigma),Ax)$ is a triple such that $Ax$ is a
  class of equations-in-context over linear $\lambda$-terms built from
  $(G,\Sigma)$.
\end{definition}
The elements of $Ax$ are called axioms (of the theory). Let $Th(Ax)$
be the smallest congruence that contains $Ax$, the equations listed in
Fig.~\ref{fig:eqs}, and that is closed under the exchange and
substitution rules.  We call the elements of $Th(Ax)$ theorems (of the
theory).
  
\begin{definition}[Models of linear $\lambda$-theories] Consider a
  linear $\lambda$-theory $((G,\Sigma),Ax)$ and an autonomous category
  $\catC$. Suppose that for each $X \in G$ we have an interpretation
  $\sem{X}$ that is a $\catC$-object and analogously for the operation
  symbols. This interpretation structure is a \emph{model} of the
  theory if all axioms are satisfied by the interpretation.
\end{definition}

Next let us turn our attention to the correspondence between linear
$\lambda$-calculus and autonomous categories, established via
soundness, completeness, and internal language theorems. Despite the
proofs of such theorems already being detailed
in~\cite{mackie93,benton92,maietti05}, we decided to briefly sketch
them below to render the presentation of some of our own results
self-contained.

\begin{theorem}[Soundness \& Completeness]~\label{theo:sound_compl}
  Consider a linear $\lambda$-theory $T$. An equation
  $\Gamma \vljud v = w : \typeA$ is a theorem of $T$ iff it
  is satisfied by all models of the theory.
\end{theorem}

\begin{proof}[Proof sketch]
  Soundness follows by induction over the rules that define $Th(Ax)$
  (Definition~\ref{defn:theory}) and by
  Theorem~\ref{theo:bsound}. Completeness is based on the idea of a
  Lindenbaum-Tarski algebra: it follows from building the
  \emph{syntactic category} $\Syn(T)$ of $T$ (also known as term
  model), showing that it possesses an autonomous structure and
  also that equality
  $\sem{\Gamma \vljud v : \typeA} = \sem{\Gamma \vljud w : \typeA}$ in
  the syntactic category is equivalent to provability
  $\Gamma \vljud v = w : \typeA$ in the theory.

  The syntactic category of $T$ has as objects the
  types of $T$ and as morphisms $\typeA \to \typeB$ the equivalence
  classes (w.r.t. provability) of terms $v$ for which we can derive
  $x : \typeA \vljud v : \typeB$.
\end{proof}

Next let us focus on the topic of internal languages, for which the
following result is quite useful.

\begin{theorem}\label{theo:classb}
  Consider a linear $\lambda$-theory $T$ and a model of $T$ on an
  autonomous category $\catC$. The model induces a functor
  $\funF : \Syn(T) \to \catC$ that (strictly) preserves the autonomous
  structure.
\end{theorem}

\begin{proof}[Proof sketch]
  Consider a model of $T$ on a category $\catC$. Then for any
  judgement $x : \typeA \vljud v : \typeB$, the induced functor
  $\funF$ sends the equivalence class $[v]$ into
  $\sem{ x : \typeA \vljud v :
    \typeB}$. 
\end{proof}

An autonomous category $\catC$ induces a linear $\lambda$-theory
$\Lang(\catC)$ whose ground types $X \in G$ are the objects of $\catC$
and whose signature $\Sigma$ of operation symbols consists of all the
morphisms in $\catC$ plus certain isomorphisms that we describe
in~\eqref{eq:iso}. The axioms of $\Lang(\catC)$ are all the equations
satisfied by the obvious interpretation in $\catC$.  In order to
explicitly distinguish the autonomous structure of $\catC$ from the
type structure of $\Lang(\catC)$ let us denote the tensor of $\catC$
by $\hat \otimes$, the unit by $\hat \typeI$, and the exponential by
$\widehat{\multimap}$. Consider then the following map on types:
\begin{flalign}\label{eq:iso}
    i(\typeI) = \hat \typeI \qquad i(X) = X \qquad i(\typeA \otimes
    \typeB) = i(\typeA)\ \hat \otimes\ i(\typeB) \qquad i(\typeA
    \multimap \typeB) = i(\typeA)\ \widehat{\multimap}\ i(\typeB)
\end{flalign}
For each type $\typeA$ we add an isomorphism $\typeA \simeq i(\typeA)$
to the theory $\Lang(\catC)$.

\begin{theorem}[Internal language]
  \label{theo:internalb}
  For every autonomous category $\catC$ there exists an equivalence of
  categories $\Syn(\Lang(\catC)) \simeq \catC$.
\end{theorem}

\begin{proof}[Proof sketch]
  By construction, we have an interpretation of $\Lang(\catC)$ in
  $\catC$ which behaves as the identity for operation symbols and
  ground types. This interpretation is a model of
  $\Lang(\catC)$ on $\catC$ and by
  Theorem~\ref{theo:classb} we obtain a functor
  $\Syn(\Lang(\catC)) \to \catC$. The functor in the opposite
  direction behaves as the identity on objects and sends a
  $\catC$-morphism $f$ into $[f(x)]$. The equivalence of categories is
  then shown by using the aforementioned isomorphisms which
  connect the type constructors of $\Lang(\catC)$ with the autonomous
  structure of $\catC$.
\end{proof}

\section{From equations to $\mathcal{V}$-equations}
\label{sec:main}

We now extend the results of the previous section to the
setting of $\mathcal{V}$-equations.

\subsection{A $\mathcal{V}$-equational deductive system}

Let $\mathcal{V}$ denote a commutative and unital quantale,
$\otimes : \mathcal{V} \times \mathcal{V} \to \mathcal{V}$ the
corresponding binary operation, and $k$ the corresponding
unit~\cite{paseka00}.  As mentioned in the introduction, $\mathcal{V}$
induces the notion of a $\mathcal{V}$-equation, \ie an equation
$t =_q s$ labelled by an element $q$ of $\mathcal{V}$. This
subsection explores this concept by introducing a
$\mathcal{V}$-equational deductive system for linear
$\lambda$-calculus and a notion of a linear
$\mathcal{V}\lambda$-theory.

Let us start by recalling two definitions concerning ordered
structures~\cite{GHK+03,JGL-topology} and then explain their relevance
to our work.

\begin{definition}
  Consider a complete lattice $L$.  For every $x, y \in L$ we say that
  $y$ is \emph{way-below} $x$ (in symbols, $y \ll x$) if for every
  subset $X \subseteq L$ whenever $x \leq \bigvee X$ there exists a
  \emph{finite} subset $A \subseteq X$ such that $y \leq \bigvee A$.
  The lattice $L$ is called \emph{continuous} iff for every $x \in L$,
  \begin{flalign*}
    x = \bigvee \{ y  \mid y \in L\ \text{and} \ y \ll x \}
  \end{flalign*}
\end{definition}

\begin{definition}
  Let $L$ be a complete lattice. A \emph{basis} $B$ of $L$ is a subset
  $B \subseteq L$ such that for every $x \in L$ the set
  $B \cap \{ y \mid y \in L\ \text{and} \ y \ll x \}$ is directed and
  has $x$ as the least upper bound.
\end{definition}
From now on we assume that the underlying lattice of $\mathcal{V}$
is continuous and has a basis
$B$ which is closed under finite joins, the multiplication of the quantale
$\otimes$ and contains the unit $k$. These assumptions will
allow us to work \emph{only} with a specified subset of
$\mathcal{V}$-equations chosen \eg for computational reasons, such as
the \emph{finite} representation of values $q \in \mathcal{V}$.
\begin{example}
  The Boolean quantale $((\{0 \leq 1\}, \vee), \otimes := \wedge)$ is
  \emph{finite} and thus continuous~\cite{GHK+03}. Since it is
  continuous, $\{0,1\}$ itself is a basis for the quantale that
  satisfies the conditions above. For the G\"{o}del
  t-norm~\cite{denecke13} $(([0,1], \vee), \otimes := \wedge)$, the
  way-below relation is the strictly-less relation $<$ with the
  exception that $0 < 0$. A basis for the underlying lattice that
  satisfies the conditions above is the set $\mathbb{Q} \cap
  [0,1]$. Note that, unlike real numbers, rationals numbers always
  have a finite representation. For the metric quantale (also known as
  Lawvere quantale) $(([0,\infty], \wedge), \otimes := +)$, the
  way-below relation corresponds to the \emph{strictly greater}
  relation with $\infty > \infty$, and a basis for the underlying
  lattice that satisfies the conditions above is the set of extended
  non-negative rational numbers. The latter also serves as basis for
  the ultrametric quantale $(([0,\infty], \wedge), \otimes := \max)$.
\end{example}
We also assume that $\mathcal{V}$ is \emph{integral}, \ie that the
unit $k$ is the top element of $\mathcal{V}$. This will allow us to
establish a smoother theory of $\mathcal{V}$-equations, whilst still
covering \eg\ all the examples above. This assumption is
common in quantale theory~\cite{velebil19}.

Recall the term formation rules of linear $\lambda$-calculus from
Fig.~\ref{fig:lang}. A $\mathcal{V}$-equation-in-context is an
expression $\Gamma \vljud v =_q w : \typeA$ with $q \in B$ (the basis
of $\mathcal{V}$), $\Gamma \vljud v : \typeA$ and
$\Gamma \vljud w : \typeA$. Let $\top$ be the top element in
$\mathcal{V}$. An equation-in-context $\Gamma \vljud v = w : \typeA$
now denotes the particular case in which both
$\Gamma \vljud v =_\top w : \typeA$ and
$\Gamma \vljud w =_\top v : \typeA$. For the case of the Boolean
quantale, $\mathcal{V}$-equations are labelled by $\{0,1\}$. We will
see that $\Gamma \vljud v =_1 w : \typeA$ can be treated as an
inequation $\Gamma \vljud v \leq w : \typeA$, whilst
$\Gamma \vljud v =_0 w : \typeA$ corresponds to a trivial
$\mathcal{V}$-equation, \ie a $\mathcal{V}$-equation that always
holds. For the G\"{o}del t-norm, we can choose $\mathbb{Q} \cap [0,1]$
as basis and then obtain what we call \emph{fuzzy inequations}. For
the metric quantale, we can choose the set of extended non-negative
rational numbers as basis and then obtain \emph{metric equations} in
the spirit of~\cite{mardare16,mardare17}. Similarly, by choosing the
ultrametric quantale $(([0,\infty], \wedge), \otimes := \max)$ with
the set of extended non-negative rational numbers as basis we obtain
what we call \emph{ultrametric equations}.

\begin{definition}[Linear $\mathcal{V}\lambda$-theories]\label{defn:vtheory}
  Consider a tuple $(G,\Sigma)$ consisting of a class $G$ of ground
  types and a class of sorted operation symbols
  $f : \typeA_1,\dots,\typeA_n \to \typeA$ with $n \geq 1$. A linear
  $\mathcal{V}\lambda$-theory $((G,\Sigma),Ax)$ is a tuple such that
  $Ax$ is a class of $\mathcal{V}$-equations-in-context over linear
  $\lambda$-terms built from $(G,\Sigma)$.
\end{definition}
\begin{figure*}[h!]
    \vspace{-16pt}
    \begin{flalign*}
      \infer[\textbf{(refl)}]{v =_\top v}{}
      \hspace{1.2cm}
      \infer[\textbf{(trans)}]{v =_{q \otimes r} u}{
        v =_q w  \qquad
        w =_r u}
      \hspace{1.2cm}
      \infer[\textbf{(weak)}]{v =_r w}{v =_q w \qquad r \leq q }
    \end{flalign*} \vspace{-0.85cm}
    \begin{flalign*}
      \hspace{2cm}
      \infer[\textbf{(arch)}]{v =_q w}{
        \forall r \ll q .\ v =_r w}
      \hspace{1.5cm}
      \infer[\textbf{(join)}]{v =_{\vee q_i} w}{\forall i \leq n.\ v =_{q_i} w}
    \end{flalign*} \hrule 
    \begin{flalign*}
      \hspace{2cm}
      \infer[]{f(v_1,\dots,v_n) =_{\otimes q_i} f(w_1,\dots,w_n)}
      {\forall i \leq n.\ v_i =_{q_i} w_i}
      \hspace{1cm}
      \infer[]{v \otimes v' =_{q \otimes r} w \otimes w'}{
        v =_q w \quad v' =_r w'}
    \end{flalign*} \vspace{-0.85cm}
    \begin{flalign*}
      \hspace{3.4cm}
      \infer[]{
      \prog{pm}\ v\ \prog{to}\ x \otimes y.\ v' =_{q \otimes r}
      \prog{pm}\ w\ \prog{to}\ x \otimes y.\ w'}
      {v =_q w \qquad v' =_r w' }
    \end{flalign*} \vspace{-0.85cm}
    \begin{flalign*}
      \infer[]{
      v\ \prog{to}\ \ast .\  v'=_{q \otimes r}
      w\ \prog{to}\ \ast .\ w'}
      {v =_q w \qquad v' =_r w'}
      \hspace{1.2cm}
      \infer[]{
      \lambda x : \typeA .\ v =_q  \lambda x :
      \typeA .\ w}{v =_q w}
      \hspace{1.2cm}
      \infer[]{v \, v' =_{q \otimes r} w \, w'}
      {v =_q w \quad v' =_r w'}
    \end{flalign*} \vspace{-0.85cm}
    \begin{flalign*} 
     \hspace{1.0cm}
      \infer{\Delta \vljud v =_q w : \typeA}{
      \Gamma \vljud v =_q w : \typeA \qquad \Delta \in
      perm(\Gamma)}
      \hspace{2.0cm}
      \infer[]{v[v'/x] =_{q \otimes r}w[w'/x]}
      {v =_q w \qquad v' =_r w'}
    \end{flalign*}
    \vspace{-26pt}
  \caption{$\mathcal{V}$-congruence rules.}
  \label{fig:theo_rules}
  \vspace{-4pt}
\end{figure*}
The elements of $Ax$ are the axioms of the theory. Let $Th(Ax)$ be the
smallest class that contains $Ax$ and that is closed under the rules
of Fig.~\ref{fig:eqs} and of Fig.~\ref{fig:theo_rules} (as usual
we omit the context and typing information). The elements of $Th(Ax)$
are  the theorems of the theory.

Let us examine the rules in Fig.~\ref{fig:theo_rules} in more
detail. They can be seen as a generalisation of the notion of a
congruence. The rules \textbf{(refl)} and \textbf{(trans)} are a
generalisation of equality's reflexivity and transitivity. Rule
\textbf{(weak)} encodes the principle that the higher the label in the
$\mathcal{V}$-equation, the `tighter' is the relation between the two
terms in the $\mathcal{V}$-equation. In other words, $v =_r w$ is
subsumed by $v =_q w$, for $r\leq q$. This can be seen clearly \eg
with the metric quantale by reading $v =_q w$ as ``the terms $v$ and
$w$ are \emph{at most} at distance $q$ from each other'' (recall that
in the metric quantale the usual order is reversed, \ie
$\leq\ :=\ \geq_{[0,\infty]}$). \textbf{(arch)} is essentially a
generalisation of the Archimedean rule
in~\cite{mardare16,mardare17}. It says that if $v =_r w$ for all
\emph{approximations} $r$ of $q$ then it is also the case that
$v =_q w$. \textbf{(join)} says that deductions are closed under finite
joins, and in particular it is always the case that $v =_\bot
w$. All other rules correspond to a generalisation of
\emph{compatibility} to a $\mathcal{V}$-equational setting.

The reader may have noticed that the rules in Fig.~\ref{fig:theo_rules}
do not contain a $\mathcal{V}$-generalisation of symmetry w.r.t.
standard equality. Such a generalisation would be:
\begin{flalign*}
  \infer{w =_q v}{v =_q w}
\end{flalign*}
This rule is not present in Fig.~\ref{fig:theo_rules} because in some
quantales $\mathcal{V}$ it forces too many
$\mathcal{V}$-equations. For example, in the Boolean quantale the condition
$v \leq w$ would automatically entail $w \leq v$ (due to
symmetry); in fact, for this particular case symmetry forces the notion
of inequation to collapse into the classical notion of equation. On
the other hand, symmetry is desirable in the (ultra)metric case
because (ultra)metrics need to respect the symmetry
equation~\cite{JGL-topology}.
\begin{definition}[Symmetric linear $\mathcal{V}\lambda$-theories]
  A symmetric linear $\mathcal{V}\lambda$-theory is a linear
  $\mathcal{V}\lambda$-theory whose set of theorems is closed under
  symmetry.
\end{definition}
In Appendix \ref{sec:simpl} we further explore how specific families
of quantales are reflected in the $\mathcal{V}$-equational system here
introduced, and briefly compare the latter to metric 
algebra~\cite{mardare16,mardare17}.

\subsection{Semantics of $\mathcal{V}$-equations}
\label{sec:vcat}
In this subsection we set the necessary background for presenting a
sound and complete class of models for (symmetric) linear
$\mathcal{V}\lambda$-theories. We start by recalling basics concepts
of $\mathcal{V}$-categories, which are central in a field initiated by
Lawvere in~\cite{lawvere73} and can be intuitively seen as generalised
metric spaces~\cite{stubbe14,hofmann20,velebil19}. As we will see,
$\mathcal{V}$-categories provide structure to suitably interpret
$\mathcal{V}$-equations.

\begin{definition}
  A (small) $\mathcal{V}$-category is a pair $(X,a)$ where $X$ is a
  class (set) and $a : X \times X \to \mathcal{V}$ is a function that
  satisfies:
  \begin{flalign*}
    k \leq a(x,x) \qquad \text{ and }  \qquad
    a(x,y) \otimes a(y,z) \leq a(x,z) \hspace{2cm}
    (x,y,z \in X)
  \end{flalign*}
  For two $\mathcal{V}$-categories $(X,a)$ and $(Y,b)$, a
  $\mathcal{V}$-functor $f : (X,a) \to (Y,b)$ is a function
  $f : X \to Y$ that satisfies the inequality
  $a(x,y) \leq b(f(x),f(y))$ for all $x,y \in X$.
\end{definition}
Small $\mathcal{V}$-categories and $\mathcal{V}$-functors form a
category which we denote by $\VCat$.  A $\mathcal{V}$-category $(X,a)$
is called \emph{symmetric} if $a(x,y) = a(y,x)$ for all $x,y \in
X$. We denote by $\VCatSy$ the full subcategory of $\VCat$ whose
objects are symmetric. Every $\mathcal{V}$-category carries a natural
order defined by $x \leq y$ whenever $k \leq a(x,y)$. A
$\mathcal{V}$-category is called \emph{separated} if its natural order
is anti-symmetric. We denote by $\VCatSe$ the full subcategory of
$\VCat$ whose objects are separated.

\begin{example}
  For $\mathcal{V}$ the Boolean quantale, $\VCatSe$ is the category
  $\Pos$ of partially ordered sets and monotone maps; 
  $\VCatSS$ is simply the category $\Set$ of sets and functions.  For
  $\mathcal{V}$ the metric quantale, $\VCatSS$ is the category $\Met$
  of extended metric spaces and non-expansive maps. In what follows we
  omit the qualifier `extended' in `extended (ultra)metric
  spaces'. For $\mathcal{V}$ the ultrametric quantale, $\VCatSS$ is
  the category of ultrametric spaces and non-expansive maps.
\end{example}
The inclusion functor $\VCatSe \hookrightarrow \VCat$ has a left
adjoint \cite{hofmann20}. It is constructed first by defining the
equivalence relation $x \sim y$ whenever $x \leq y$ and $y \leq x$
(for $\leq$ the natural order introduced earlier). Then this relation
induces the separated $\mathcal{V}$-category $(X/_\sim, \tilde a)$
where $\tilde a$ is defined as $\tilde a([x],[y]) = a(x,y)$ for every
$[x],[y] \in X/_\sim$. The left adjoint of the inclusion functor
$\VCatSe \hookrightarrow \VCat$ sends every $\mathcal{V}$-category
$(X,a)$ to $(X/_\sim, \tilde a)$. This quotienting construct preserves
symmetry, and therefore we automatically obtain the following result.
\begin{theorem}\label{theo:reflective}
  The inclusion functor $\VCatSS \hookrightarrow \VCatSy$ has a left
  adjoint.
\end{theorem}
% \begin{corollary}
%   Both inclusion functors $\Ord \hookrightarrow \POrd$ and
%   $\Met \hookrightarrow \PMet$ have a left adjoint.
% \end{corollary}
Next, we recall notions of enriched category theory~\cite{kelly82}
instantiated into the setting of \emph{autonomous categories enriched
  over $\mathcal{V}$-categories}.  We will use the enriched structure
to give semantics to $\mathcal{V}$-equations between linear
$\lambda$-terms.
First, note that every category $\VCat$ is autonomous with the tensor
$(X,a) \otimes (Y,b) := (X \times Y, a \otimes b)$ where $a\otimes b$ is defined as
%\begin{flalign*}
  $(a \otimes b)((x,y), (x',y')) = a(x,x') \otimes b(y,y')$
%\end{flalign*}
and the set of $\mathcal{V}$-functors $\VCat((X,a),(Y,b))$ equipped
with the map
$(f,g) \mapsto \bigwedge_{x \in X} b(f(x),g(x))$.

\begin{theorem}\label{theo:auto}
  The categories $\VCatSy$,
  $\VCatSe$, and $\VCatSS$ inherit the autonomous structure of
  $\VCat$ whenever $\mathcal{V}$ is integral.
\end{theorem}
Since we assume that $\mathcal{V}$ is integral, this last theorem
allows us to formally define the notion of categories enriched over
$\mathcal{V}$-categories using~\cite{kelly82}.

\begin{definition}\label{def:VCatEnriched}
  A category $\catC$ is $\VCat$-enriched (or simply, a
  $\VCat$-category) if for all $\catC$-objects $X$ and $Y$ the hom-set
  $\catC(X,Y)$ is a $\mathcal{V}$-category and if the composition 
  of $\catC$-morphisms, 
  \begin{flalign*}
    (\ \cdot\ ) : \catC(X,Y) \otimes \catC(Y,Z)
    \longrightarrow \catC(X,Z)
  \end{flalign*}
  is a $\mathcal{V}$-functor. Given two $\VCat$-categories $\catC$ and
  $\catD$ and a functor $\funF : \catC \to \catD$, we call $\funF$ a
  $\VCat$-functor if for all $\catC$-objects $X$ and $Y$ the map
  $\funF_{X,Y} : \catC(X,Y) \to \catD(\funF X, \funF, Y)$ is a
  $\mathcal{V}$-functor.  An adjunction
  $\catC : \funF \dashv \funG : \catD$ is called $\VCat$-enriched if
  the underlying functors $\funF$ and $\funG$ are $\VCat$-functors and
  if for all objects $X \in |\catC|$ and $Y \in |\catD|$ there exists
  a $\mathcal{V}$-isomorphism 
  $\catD(\funF X, Y) \simeq \catC(X, \funG Y)$
  natural in $X$ and $Y$.
\end{definition}
If $\catC$ is a $\VCat$-category then $\catC \times \catC$ is also a
$\VCat$-category via the tensor operation $\otimes$ in $\VCat$.  We
take advantage of this fact in the following definition.
\begin{definition}\label{defn:enr_aut}
  A $\VCat$-enriched autonomous category $\catC$ is an
  autonomous and $\VCat$-category $\catC$ such that the
  bifunctor $\otimes : \catC \times \catC \to \catC$ is a
  $\VCat$-functor and the adjunction
  $(- \otimes X) \dashv (X \multimap -)$ is a $\VCat$-adjunction.
\end{definition}

\begin{example}\label{ex:pos_met}
  Recall that $\Pos \simeq \VCatSe$ when $\mathcal{V}$ is the Boolean
  quantale. According to Theorem~\ref{theo:auto} the category $\Pos$
  is autonomous. It follows by general results that the category is
  $\Pos$-enriched~\cite{Bor94a}. It is also easy to see that its
  tensor is $\Pos$-enriched and that the adjunction
  $(- \otimes X) \dashv (X \multimap -)$ is
  $\Pos$-enriched. Therefore, $\Pos$ is an instance of
  Definition~\ref{defn:enr_aut}. Note also that $\Set \simeq \VCatSS$
  for $\mathcal{V}$ the Boolean quantale and that $\Set$ is an
  instance of Definition~\ref{defn:enr_aut}.

  Recall that $\Met \simeq \VCatSS$ when $\mathcal{V}$ is the metric
  quantale.  Thus, the category $\Met$ is autonomous
  (Theorem~\ref{theo:auto}) and $\Met$-enriched~\cite{Bor94a}.  It
  follows as well from routine calculations that its tensor is
  $\Met$-enriched and that the adjunction
  $(- \otimes X) \dashv (X \multimap -)$ is $\Met$-enriched. Therefore
  $\Met$ is an instance of Definition~\ref{defn:enr_aut}. An analogous
  reasoning tells that the category of ultrametric spaces (enriched
  over itself) is also an instance of Definition~\ref{defn:enr_aut}.
\end{example}
Finally, recall the interpretation of linear $\lambda$-terms on an
autonomous category $\catC$ (Section~\ref{sec:internal}) and assume
that $\catC$ is $\VCat$-enriched. Then we say that a
$\mathcal{V}$-equation $\Gamma \vljud v =_q w : \typeA$ is
\emph{satisfied} by this interpretation if
$a(\sem{\Gamma \vljud v : \typeA},\sem{\Gamma \vljud w : \typeA}) \geq
q$ where
$a : \catC(\sem{\Gamma},\sem{\typeA}) \times
\catC(\sem{\Gamma},\sem{\typeA}) \to \mathcal{V}$ is the underlying
function of the $\mathcal{V}$-category
$\catC(\sem{\Gamma},\sem{\typeA})$.

\begin{theorem}\label{theo:sound}
  The rules listed in Fig.~\ref{fig:eqs} and Fig.~\ref{fig:theo_rules}
  are sound for $\VCat$-enriched autonomous categories $\catC$.
  Specifically, if $\Gamma \vljud v =_q w : \typeA$ results from the
  rules in Fig.~\ref{fig:eqs} and Fig.~\ref{fig:theo_rules} then
  $a(\sem{\Gamma \vljud v : \typeA},\sem{\Gamma \vljud w : \typeA})
  \geq q$.
\end{theorem}

\begin{proof}[Proof of Theorem~\ref{theo:sound}]
  Let us focus first on the equations listed in Fig.~\ref{fig:eqs}.
  Recall that an equation $\Gamma \vljud v = w : \typeA$ abbreviates
  the $\mathcal{V}$-equations $\Gamma \vljud v =_\top w : \typeA$ and
  $\Gamma \vljud w =_\top v : \typeA$.  Moreover, we already know that
  the equations listed in Fig.~\ref{fig:eqs} are sound for
  autonomous categories, specifically if $v = w$ is an equation of
  Fig.~\ref{fig:eqs} then $\sem{v} = \sem{w}$ in $\catC$
  (Theorem~\ref{theo:bsound}).  Thus, by the definition of a
  $\mathcal{V}$-category and by the assumption of $\mathcal{V}$
  being integral ($k = \top$) we obtain
  $a(\sem{v},\sem{w}) \geq k = \top$ and
  $a(\sem{w},\sem{v}) \geq k = \top$.

  Let us now focus on the rules listed in
  Fig.~\ref{fig:theo_rules}. The first three rules follow from the
  definition of a $\mathcal{V}$-category and the transitivity property
  of $\leq$. Rule \textbf{(arch)} follows from the continuity of
  $\mathcal{V}$, specifically from the fact that $q$ is the
  \emph{least} upper bound of all elements $r$ that are way-below $q$.
  Rule \textbf{(join)} follows from the definition of least upper bound.
  The remaining rules follow from the definition of the tensor functor
  $\otimes$ in $\VCat$, the fact that $\catC$ is $\VCat$-enriched,
  $\otimes : \catC \times \catC \to \catC$ is a $\VCat$-functor, and
  the fact that $(- \otimes X) \dashv (X \multimap -)$ is a
  $\VCat$-adjunction. For example, for the sixth rule we reason as
  follows:
  \begin{align*}
    & \, a(\sem{f(v_1,\dots,v_n)}, \sem{f(w_1,\dots,w_n)}) \\
    & = a(\sem{f} \comp (\sem{v_1} \otimes \dots \otimes \sem{v_n})
    \comp \spl_{\Gamma_1;\dots;\Gamma_n} \comp \sh_E,
    \sem{f} \comp (\sem{w_1} \otimes \dots \otimes \sem{w_n}) \comp
    \spl_{\Gamma_1;\dots;\Gamma_n} \comp \sh_E) & \\
    & \geq a(\sem{f} \comp (\sem{v_1} \otimes \dots \otimes \sem{v_n}),
    \sem{f} \comp (\sem{w_1} \otimes \dots \otimes \sem{w_n})) & \\
    & \geq a(\sem{v_1} \otimes \dots \otimes \sem{v_n}),
      (\sem{w_1} \otimes \dots \otimes \sem{w_n}) & \\
    & \geq a(\sem{v_1},\sem{w_1}) \otimes \dots \otimes
      a(\sem{v_n},\sem{w_n}) &  \\
    & \geq q_1 \otimes \dots \otimes q_n &  
  \end{align*}
  where the second step follows from the fact that
  $\spl_{\Gamma_1;\dots;\Gamma_n} \comp \sh_E$ is a morphism in
  $\catC$ and that $\catC$ is $\VCat$-enriched. The third step follows
  from an analogous reasoning. The fourth step follows from the fact that
  $\otimes : \catC \times \catC \to \catC$ is a $\VCat$-functor. The
  last step follows from the premise of the rule in question.  As
  another example, the proof for the substitution rule  proceeds similarly:
  \begin{align*}
    & \, a(\sem{v[v'/x]}, \sem{w[w'/x]}) \\
    & = a(\sem{v} \comp \join_{\Gamma,\typeA} \comp
      (\id \otimes \sem{v'}) \comp \spl_{\Gamma;\Delta},
      \sem{w} \comp \join_{\Gamma,\typeA} \comp
      (\id \otimes \sem{w'}) \comp \spl_{\Gamma;\Delta}) & \\
    & \geq a(\sem{v} \comp \join_{\Gamma,\typeA} \comp
      (\id \otimes \sem{v'}),
      \sem{w} \comp \join_{\Gamma,\typeA} \comp
      (\id \otimes \sem{w'})) & \\
    & \geq
    a(\id \otimes \sem{v'}, \id \otimes \sem{w'}) \otimes
    a(\sem{v} \comp \join_{\Gamma,\typeA},
    \sem{w} \comp \join_{\Gamma,\typeA}) \\
    & \geq
    a(\id \otimes \sem{v'}, \id \otimes \sem{w'}) \otimes
      a(\sem{v}, \sem{w}) \\
    & \geq
    a(\id,\id) \otimes a(\sem{v'}, \sem{w'})  \otimes
    a(\sem{v}, \sem{w}) \\
    & =
    a(\sem{v'}, \sem{w'}) \otimes a(\sem{v}, \sem{w}) \\
    & \geq q \otimes r
  \end{align*}
  The proof for the rule concerning $(\multimap_i)$ additionally
  requires the following two facts: if a $\mathcal{V}$-functor
  $f : (X,a) \to (Y,b)$ is an isomorphism then
  $a(x,x') = b(f(x),f(x'))$ for all $x,x' \in X$. For a context
  $\Gamma$, the morphism
  $\join_{\Gamma; x : \typeA} : \sem{\Gamma} \otimes \sem{\typeA} \to
  \sem{\Gamma, x : \typeA}$ is an isomorphism in $\catC$. The proof
  for the rule concerning the permutation of variables (exchange) also
  makes use of the fact that $\sem{\Delta} \to \sem{\Gamma}$ is
  an isomorphism.
\end{proof}

\subsection{Soundness, completeness, and internal language}

In this subsection we establish a formal correspondence between linear
$\mathcal{V}\lambda$-theories and $\VCat$-enriched autonomous
categories, via soundness, completeness, and internal language
theorems. A key construct in this correspondence is the quotienting of
a $\mathcal{V}$-category into a \emph{separated}
$\mathcal{V}$-category: we will use it to identify linear
$\lambda$-terms when generating a syntactic category (from a linear
$\mathcal{V}\lambda$-theory) that satisfies the axioms of
autonomous categories. This naturally leads to the following notion of
a model for linear $\mathcal{V}\lambda$-theories.
\begin{definition}[Models of linear $\mathcal{V}\lambda$-theories]
  Consider a linear $\mathcal{V}\lambda$-theory $((G,\Sigma),Ax)$ and
  a $\VCatSe$-enriched autonomous category $\catC$. Suppose that for
  each $X \in G$ we have an interpretation $\sem{X}$ as a
  $\catC$-object and analogously for the operation symbols. This
  interpretation structure is a model of the theory if all axioms in
  $Ax$ are satisfied by the interpretation.
\end{definition}
Another thing that we need to take into account is the size of
categories.  In Section~\ref{sec:internal} we did not assume that
autonomous categories should be locally small. In particular linear
$\lambda$-theories are able to generate non-(locally small)
categories. Now we need to be stricter because $\VCatSe$-enriched
autonomous categories are always locally small (recall the definition
of $\VCatSe$).  Thus for two types $\typeA$ and $\typeB$ of a
$\mathcal{V}\lambda$-theory $T$, consider the class
$\closValBP{\typeA}{\typeB}$ of values $v$ such that
$x : \typeA \vljud v : \typeB$. We equip $\closValBP{\typeA}{\typeB}$
with the function
$a : \closValBP{\typeA}{\typeB} \times \closValBP{\typeA}{\typeB} \to
\mathcal{V}$ defined by,
\begin{flalign*}
  a(v,w) = \bigvee \{ q \mid v =_q w \text{ is a theorem of } T\}
\end{flalign*}
It is easy to see that $(\closValBP{\typeA}{\typeB},a)$ is a (possibly
large) $\mathcal{V}$-category. We then quotient this
$\mathcal{V}$-category into a \emph{separated} $\mathcal{V}$-category
which we suggestively denote by $\catC(\typeA,\typeB)$ (as detailed in
the proof of the next theorem, $\catC(\typeA,\typeB)$ will serve as a
hom-object of a syntactic category $\catC$ generated from a linear
$\mathcal{V}\lambda$-theory). Following the nomenclature
of~\cite{linton66}, we call $T$ \emph{varietal} if
$\catC(\typeA,\typeB)$ is a \emph{small} $\mathcal{V}$-category. In
the rest of the paper we will only work with varietal theories and
locally small categories.

\begin{theorem}[Soundness \& Completeness]\label{theo:sound_compl2}
  Consider a varietal $\mathcal{V}\lambda$-theory. A
  $\mathcal{V}$-equation-in-context $\Gamma \vljud v =_q w : \typeA$
  is a theorem iff it holds in all models of the theory.
\end{theorem}

\begin{proof}[Proof sketch]
  Soundness follows by induction over the rules that define the class
  $Th(Ax)$ (Definition~\ref{defn:vtheory}) and by
  Theorem~\ref{theo:sound}.  For completeness, we use a strategy
  similar to the proof of
  Theorem~\ref{theo:sound_compl}, and take advantage of the
  quotienting of a $\mathcal{V}$-category into a separated
  $\mathcal{V}$-category. Recall that we assume that the theory is
  \emph{varietal} and therefore can safely take $\catC(\typeA,\typeB)$
  to be a small $\mathcal{V}$-category. Note that the quotienting
  process identifies all terms $x : \typeA \vljud v : \typeB$ and
  $x : \typeA \vljud w : \typeB$ such that $v =_\top w$ and
  $w =_\top v$. Such a relation contains the equations-in-context from
  Fig.~\ref{fig:eqs} and moreover it is straighforward to show that
  it is compatible with the term formation rules of linear
  $\lambda$-calculus (Fig.~\ref{fig:lang}). So, analogously to
  Theorem~\ref{theo:sound_compl} we obtain an autonomous category
  $\catC$ whose objects are the types of the language and whose
  hom-sets are the underlying sets of the $\mathcal{V}$-categories
  $\catC(\typeA,\typeB)$.

  Our next step is to show that the category $\catC$ has a
  $\VCatSe$-enriched autonomous structure.  We start by showing that the
  composition map
  $\catC(\typeA,\typeB) \otimes \catC(\typeB,\typeC) \to
  \catC(\typeA,\typeC)$ is a $\mathcal{V}$-functor:
  \begin{align*}
    a(([v'],[v]), ([w'],[w]))
    & = a([v],[w]) \otimes a([v'],[w'])  & \\
    & = a(v,w) \otimes a(v',w') & \\
    & = \bigvee \{ q \mid v =_q w \} \otimes
      \bigvee \{ r \mid v' =_r w' \}
    &  \\
    & = \bigvee \{ q \otimes r \mid v =_q w, v' =_r w' \}
    &  \\
    & \leq \bigvee \{ q \mid v[v'/x]
      =_q w[w'/x] \}  &
    (A \subseteq B \Rightarrow \bigvee A \leq \bigvee B) \\
    & = a( v[v'/x] , w[w'/x]) & \\
    & = a( [v[v'/x]] , [w[w'/x]]) \\
    & = a([v] \comp [v'] ,[w] \comp [w'])
  \end{align*}
  The fact that $\otimes : \catC \times \catC \to \catC$ is a
  $\VCat$-functor follows by an analogous reasoning. Next, we need to
  show that $(- \otimes X) \dashv (X \multimap -)$ is a
  $\VCat$-adjunction. It is straightforward to show that both
  functors are $\VCat$-functors, and from a similar reasoning it follows
  that the isomorphism
  $\catC(\typeB, \typeA \multimap \typeC) \simeq \catC(\typeB \otimes
  \typeA, \typeC)$ is a $\mathcal{V}$-isomorphism.

  The final step is to show that if an equation
  $\Gamma \vljud v =_q w : \typeA$ with $q \in B$ is satisfied by
  $\catC$ then it is a theorem of the linear
  $\mathcal{V}\lambda$-theory. By assumption
  $a([v],[w]) = a(v,w) = \bigvee \{ r \mid v =_r w \} \geq q$. It follows
  from the definition of the way-below relation that for all $x \in B$
  with $x \ll q$ there exists a \emph{finite} set
  $A \subseteq \{ r \mid v =_r w \}$ such that $x \leq \bigvee A$. Then
  by an application of rule \textbf{(join)} in Fig.~\ref{fig:theo_rules}
  we obtain $v =_{\bigvee A} w$, and consequently rule \textbf{(weak)} in
  Fig.~\ref{fig:theo_rules} provides $v =_x w$ for all $x \ll
  q$. Finally, by an application of rule \textbf{(arch)} in
  Fig.~\ref{fig:theo_rules}  we deduce that
  $v =_q w$ is part of the theory.
\end{proof}

Next we establish results that will be key in the proof of the
internal language theorem.  Let $\Syn(T)$ be syntactic category of a
linear $\mathcal{V}\lambda$-theory $T$, as described in
Theorem~\ref{theo:sound_compl2}.
\begin{theorem}\label{theo:class}
  Consider a linear $\mathcal{V}\lambda$-theory $T$ and a model of $T$
  on a $\VCatSe$-enriched autonomous category $\catC$. The model
  induces a $\VCatSe$-functor $\Syn(T) \to \catC$ that (strictly)
  preserves the autonomous structure of $\Syn(T)$.
\end{theorem}

\begin{proof}
  Consider a model of $T$ over $\catC$. Let $a$ denote the underlying
  function of the hom-($\mathcal{V}$-categories) in $\Syn(T)$ and $b$
  the underlying function of the hom-($\mathcal{V}$-categories) in
  $\catC$. Then note that if $[v] = [w]$ then, by completeness, the
  equations $v =_\top w$ and $w =_\top v$ are theorems, which means
  that $\sem{v} = \sem{w}$ by the definition of a model and
  \emph{separability}. This allows us to define a mapping
  $F : \Syn(T) \to \catC$ that sends each type $\typeA$ to
  $\sem{\typeA}$ and each morphism $[v]$ to $\sem{v}$. The fact that
  this mapping is an autonomous functor follows from an analogous
  reasoning to the one used in the proof of Theorem~\ref{theo:classb}.
  We now need to show that this functor is $\VCatSe$-enriched. Recall that $a([v],[w]) = \bigvee \{ q \mid v =_q w \}$ and observe
  that for every $v =_q w$ in the previous quantification we have
  $b(\sem{v},\sem{w}) \geq q$ (by the definition of a model), which
  establishes, by the definition of a least upper bound,
  $a([v],[w]) = \bigvee \{ q \mid v =_q w \} \leq b(\sem{v},\sem{w})$.
\end{proof}

Consider now a $\VCatSe$-enriched autonomous category $\catC$. It
induces a linear $\mathcal{V}\lambda$-theory $\Lang(\catC)$ whose
ground types and operations symbols are defined as in the case of
linear $\lambda$-theories (recall Section~\ref{sec:internal}). The axioms
of $\Lang(\catC)$ are all the $\mathcal{V}$-equations-in-context that
are satisfied by the obvious interpretation on $\catC$.

\begin{theorem}\label{theo:vari}
  The linear $\mathcal{V}\lambda$-theory $\Lang(\catC)$ is varietal.
\end{theorem}
In conjunction with the proof of Theorem~\ref{theo:sound_compl2}, a
consequence of this last theorem is that $\Syn(\Lang(\catC))$ is a
$\VCatSe$-enriched category. Then we state,
\begin{theorem}[Internal language]\label{theo:qinternal}
  For every $\VCatSe$-enriched autonomous category $\catC$ there exists a
  $\VCatSe$-equivalence of categories
  $\Syn(\Lang(\catC)) \simeq \catC$.
\end{theorem}

\begin{proof}
  Let $a$ denote the underlying function of the
  hom-($\mathcal{V}$-categories) in $\Syn(\Lang(\catC))$ and $b$ the
  underlying function of the hom-($\mathcal{V}$-categories) in
  $\catC$.  We have, by construction, a model of $\Lang(\catC)$ on
  $\catC$ which acts as the identity in the interpretation of ground
  types and operation symbols. We can then appeal to
  Theorem~\ref{theo:class} to establish a $\VCatSe$-functor
  $\Syn(\Lang(\catC)) \to \catC$. Next, the functor working on the
  inverse direction behaves as the identity on objects and sends a
  morphism $f$ into $[f(x)]$. Let us show that it is
  $\VCatSe$-enriched. First, observe that if $q \ll b(f,g)$ in $\catC$
  and $q \in B$ then $f(x) =_q g(x)$ is a theorem of $\Lang(\catC)$,
  due to the fact that $\ll$ entails $\leq$ and by the definition of
  $\Lang(\catC)$. Using the definition of a basis, we thus obtain
  $b(f,g) = \bigvee \{ q \in B \mid \> q \ll b(f,g) \} \leq \bigvee \{
  q \in B \mid f(x) =_q g(x) \} = a([f(x)],[g(x)])$. The equivalence
  of categories is then shown as in the proof of
  Theorem~\ref{theo:internalb}.
\end{proof}

All the results in this section can be extended straightforwardly to the
case of \emph{symmetric} linear $\mathcal{V}\lambda$-theories and
$\VCatSS$-enriched autonomous categories.

\section{Examples of linear $\mathcal{V}\lambda$-theories and their models}
\label{sec:examples}

\begin{example}[Wait calls]\label{ex:waitcalls}
  We now return to the example of wait calls and the corresponding
  metric axioms~\eqref{ax} sketched in the Introduction. Let us build
  a model over $\Met$ for this theory: fix a metric space $A$,
  interpret the ground type $X$ as $\Nats \otimes A$ and the operation
  symbol $\prog{wait_n} : X \to X$ as the non-expansive map, $
    \sem{\prog{wait_n}} : \Nats \otimes A \to \Nats \otimes A$
    , $(i,a) \mapsto (i + n, a)$.
  Since we already know that $\Met$ is enriched over itself (recall
  Definition~\ref{defn:enr_aut} and Example~\ref{ex:pos_met}) we only
  need to show that the axioms in~\eqref{ax} are satisfied by the
  proposed interpretation. This can be shown via a few routine
  calculations.

  Now, it may be the case that is unnecessary to know the
  \emph{distance} between the execution time of two programs --
  instead it suffices to know whether a program finishes its execution
  \emph{before} another one.  This leads us to linear
  $\mathcal{V}\lambda$-theories where $\mathcal{V}$ is the Boolean
  quantale. We call such theories \emph{linear ordered
    $\lambda$-theories}.  Recall the language from the Introduction
  with a single ground type $X$ and the signature of wait calls
  $\Sigma = \{ \prog{wait_n} : X \to X \mid n \in \Nats \}$. Then we
  adapt the metric axioms \eqref{ax} to the case of the Boolean
  quantale by considering instead:
  \begin{flalign*}
    \prog{wait_0}(x) = x \hspace{1cm} \prog{wait_n}(\prog{wait_m}(x)) =
    \prog{wait_{n + m }}(x) \hspace{1cm}
    \infer{\prog{wait_n}(x) \leq \prog{wait_m}(x)}{n \leq m}
  \end{flalign*}
  where a classical equation $v = w$ is shorthand for $v \leq w$ (\ie
  $v =_1 w$) and $w \leq v$ (\ie $w =_1 v$).  In the resulting theory
  we can consider for instance (and omitting types for simplicity) the
  $\lambda$-term that defines the composition of two functions
  $\lambda f.\ \lambda g.\ g\ (f\ x)$, which we
  denote by $v$, and show that:
  \begin{flalign*} 
    v\ (\lambda x.\ \prog{wait_1}(x)) \leq v\
    (\lambda x.\ \prog{wait_1}(\prog{wait_1}(x)))
  \end{flalign*}
  This inequation between higher-order programs arises from the
  argument $\lambda x.\ \prog{wait_1}(\prog{wait_1}(x))$ being
  costlier than the argument $\lambda x.\ \prog{wait_1}(x)$ --
  specifically, the former will invoke one more wait call
  ($\prog{wait_1}$) than the latter. Moreover, the inequation entails
  that for every argument $g$ the execution time of computation
  $v\ (\lambda x.\ \prog{wait_1}(x))\ g$ will always be smaller than
  that of computation
  $v\ (\lambda x.\ \prog{wait_1}(\prog{wait_1}(x)))\ g$ since it
  invokes one more wait call. Thus in general the inequation tells
  that costlier programs fed as input to $v$ will result in longer
  execution times when performing the corresponding computation.
  In order to build a model for the ordered theory of wait calls, we
  consider a poset $A$ and define a model over $\Pos$ by sending $X$
  into $\Nats \otimes A$ and $\prog{wait_n} : X \to X$ to the monotone
  map $ \sem{\prog{wait_n}} : \Nats \otimes A \to \Nats \otimes A$,
  $(i,a) \mapsto (i + n, a)$.  Since we already know that $\Pos$ is
  enriched over itself (recall Definition~\ref{defn:enr_aut} and
  Example~\ref{ex:pos_met}) we only need to show that the ordered
  axioms are satisfied by the proposed interpretation. But again, this
  can be shown via a few routine calculations.
\end{example}

\begin{example}[Probabilistic programs]
  We consider ground types $\real,\real^+,\unit$ and a signature
  consisting of
  $\{ r : \typeI \to \real \mid r \in \mathbb{Q} \} \cup \{ r^+ :
  \typeI \to \real^+ \mid r \in \mathbb{Q}_{\geq 0}\} \cup \{ r^u :
  \typeI \to \unit \mid r \in [0,1] \cap \mathbb{Q} \}$, an operation
  $+$ of type $\real,\real\to\real$, and \emph{sampling} functions
  $\bern: \real,\real,\unit\to\real$ and
  $\normal:\real,\real^+\to\real$. Whenever no ambiguities arise, we
  drop the superscripts in $r^u$ and $r^+$.  Operationally,
  $\bern(x,y,p)$ generates a sample from the Bernoulli distribution
  with parameter $p$ on the set $\{x,y\}$, whilst $\normal(x,y)$
  generates a normal deviate with mean $x$ and standard deviation $y$.
  We then postulate the metric axiom,
\begin{flalign}
\infer{\bern(x_1,x_2,p(\ast)) =_{\abs{p-q}} \bern(x_1,x_2,q(\ast))}
	{p, q\in [0,1] \cap \mathbb{Q}} \label{eq:bern}
\end{flalign}
We interpret the resulting linear metric $\lambda$-theory in the
category $\Ban$ of Banach spaces and short operators, \ie the
semantics of \cite{DK20a,K81c} without the order structure needed to
interpret \texttt{while} loops.  This is the usual 
representation of Markov chains/kernels as matrices/operators.

\begin{theorem}\label{theo:Ban}
  The category $\Ban$ is a $\Met$-enriched autonomous category, and
  thus an instance of Definition~\ref{defn:enr_aut}.
\end{theorem}
In particular, $\Ban$ forms a model for the theory of our small
probabilistic language via the following interpretation.  We define
$\sem{\real}=\meas\R$, the Banach space of finite Borel measures on
$\R$ equipped with the total variation norm, and similarly
$\sem{\real^+}=\meas\R^+$ and $\sem{\unit}=\meas[0,1]$. We have
$\sem{\typeI}=\R \ni 1$, and for every $r\in\mathbb{Q}$ we put
$\sem{r}:\R\to\meas\R,x\mapsto x\delta_r$, where $\delta_r$ is the
Dirac delta over $r$; thus $\sem{r}(1)=\delta_r$. We define an
analogous interpretation for the operation symbols $r^+$ and
$r^u$. For $\mu, \upsilon \in\meas\R$ we define
$\sem{+}(\mu \otimes \upsilon) \defeq +_{\ast}(\mu\otimes\upsilon)$
the pushforward under $+$ of the product measure
$\mu \otimes \upsilon$ (seen as an element of $\meas\R\otimes\meas\R$,
see \cite{DK20a}). For $\mu,\upsilon,\xi \in\meas\R$ we define
$\sem{\bern} (\mu \otimes \upsilon \otimes \xi)\defeq
\mathrm{bern}_{\ast}(\mu \otimes \upsilon \otimes \xi)$, the
pushforward of the product measure $\mu \otimes \upsilon \otimes \xi$
under the Markov kernel
$\mathrm{bern}: \R^3\to \R, (u,v,p)\mapsto p\delta_u + (1-p)\delta_v$,
and similarly for $\sem{\normal}$ (see \cite{DK20a} for the definition
of pushforward by a Markov kernel).

This interpretation is sound (a proof is given in the Appendix)
because the norm on $\meas\R$ is the total variation norm, and the
metric axiom~\eqref{eq:bern} describes the total variation distance
between the corresponding Bernoulli distributions. Consider now the
following $\lambda$-terms (where we abbreviate the constants
$0(\ast),1(\ast),p(\ast),q(\ast)$ to $0,1,p,q$, respectively),
\begin{flalign*}
& \mathtt{walk1}\defeq \lambda x:\real. \bern(0,x+\normal(0,1),p)
\\
& \mathtt{walk2}\defeq \lambda x:\real. \bern(0,x+\normal(0,1),q), \qquad p,q\in[0,1] \cap \mathbb{Q}.
\end{flalign*}
As the names suggest, these two terms of type $\real\multimap\real$
are denoted by random walks on $\R$. At each call,
$\mathtt{walk1}$ (\resp $\mathtt{walk2}$) performs a jump drawn
randomly from a standard normal distribution, or is forced to return
to the origin with probability $p$ (\resp $q$).  These are
non-standard random walks whose semantics are concretely given by
complicated operators $\meas\R\to\meas\R$, but the simple quantitative
equational system of Fig.~\ref{fig:theo_rules} and the axiom
\eqref{eq:bern} allow us to easily derive
$\mathtt{walk1} = _{\abs{p-q}}\mathtt{walk2}$ without having to
compute the semantics of these terms. In other words, the soundness of
\eqref{eq:bern} is enough to tightly bound the distance between two
non-trivial random walks represented as higher-order terms in a
probabilistic programming language. Furthermore, the 
tensor in the $\lambda$-calculus allows us to easily scale up this reasoning
to random walks in higher dimensions such as
$\mathtt{walk1}\otimes\mathtt{walk2}$ on $\R^2$.
\end{example}

\section{Conclusions and future work}
\label{sec:concl}
We introduced the notion of a $\mathcal{V}$-equation which generalises
the well-established notions of equation,
inequation~\cite{kurz2017quasivarieties,adamek20}, and metric
equation~\cite{mardare16,mardare17}. We then presented a sound and
complete $\mathcal{V}$-equational system for linear
$\lambda$-calculus, illustrated with different examples of programs
containing real-time and probabilistic behaviour.

\noindent
\textbf{Functorial connection to previous work.}  As a concluding
note, let us introduce a simple yet instructive functorial connection
between (1) the categorical semantics of linear $\lambda$-calculus
with the $\mathcal{V}$-equational system, (2) the categorical
semantics of linear $\lambda$-calculus with the equational system of
Section~\ref{sec:internal}, and (3) the algebraic semantics of the
exponential free, multiplicative fragment of linear logic. First we
need to recall some well-known facts. As detailed before, typical
categorical models of linear $\lambda$-calculus and its equational
system are locally small autonomous categories. The latter form a
quasicategory $\Aut$ whose morphisms are autonomous functors. The
usual algebraic models of the exponential free, multiplicative
fragment of linear logic are the so-called
\emph{lineales}~\cite{paiva99}. In a nutshell, a lineale is a poset
$(X,\leq)$ paired with a commutative, monoid operation
$\otimes : X \times X \to X$ that satisfies certain
conditions. Lineales are almost quantales: the only difference is that
they do not require $X$ to be cocomplete. The key idea in algebraic
semantics is that the order $\leq$ in the lineale encodes the logic's
entailment relation. A functorial connection between autonomous
categories and lineales (\ie\ between (2) and (3)) is stated
in~\cite{paiva99} and is based on the following two
observations. First, (possibly large) lineales can be 
seen as \emph{thin} autonomous categories, \ie\ as elements of the
enriched quasicategory $\{0,1\}$-$\Aut$. Second, the inclusion
$\{0,1\}$-$\Aut$ $\hookrightarrow \Aut$ has a left adjoint which
\emph{collapses all morphisms} of a given autonomous category $\catC$
(intuitively, it eliminates the ability of $\catC$ to differentiate
different terms between two types). This provides an adjoint situation
between (2) and (3). We can now expand this connection to our
categorical semantics of linear $\lambda$-calculus and corresponding
$\mathcal{V}$-equational system (\ie\ (1)) in the following way. The
forgetful functor $\VCat \to \Set$ has a left adjoint
$D : \Set \to \VCat$ which sends a set $X$ to $DX = (X,d)$,
\begin{flalign*}
  d(x_1,x_2) = \begin{cases}
    k & \text{ if } x_1 = x _2 \\
    \bot & \text{otherwise} \end{cases}
\end{flalign*}
This left adjoint is strong monoidal, specifically we have
$D(X_1 \times X_2) = DX_1 \otimes DX_2$ and
$\typeI = (1,(\ast,\ast) \mapsto k) = D1$. This gives rise to the
functors,
\[
\begin{tikzcd}
  (\VCat)\text{-} \Aut \arrow[r, shift right=1ex]
  & \Aut \arrow[l, "\hat D"', shift right=1ex]
  \arrow[phantom, l, "\scriptscriptstyle{\bot}"]
  \arrow[r, "c", shift left=1ex]
  & \{0,1\}\text{-}\Aut
  \arrow[phantom, l, "\scriptscriptstyle{\bot}"]
  \arrow[l, shift left=1ex]
\end{tikzcd}
\]
where $\hat D$ equips the hom-sets of an autonomous category with the
corresponding discrete $\mathcal{V}$-category and $c$ collapses all
morphisms of an autonomous category as described earlier. The right
adjoint of $\hat D$ forgets the $\mathcal{V}$-categorical structure
between terms (\ie\ morphisms) and the right adjoint of $c$ is the
inclusion functor mentioned earlier. Note that $\hat D$ restricts to
$(\VCatSe)$-$\Aut$ and $(\VCatSS)$-$\Aut$, and thus we obtain a
functorial connection between the categorical semantics of linear
$\lambda$-calculus with the $\mathcal{V}$-equational system (\ie\
(1)), (2), and (3). In essence, the connection formalises the fact
that our categorical models admit a richer structure over terms (\ie\
morphisms) than the categorical models of linear $\lambda$-calculus
and its classical equational system. The latter in turn permits the
existence of different terms between two types as opposed to the
algebraic semantics of the exponential free, multiplicative fragment
of linear logic. The connection also shows that models for (2) and (3)
can be mapped into models of our categorical semantics by equipping
the respective hom-sets with a trivial, discrete structure.

\noindent
\textbf{Future work.}  Recall that linear $\lambda$-calculus is at the
root of different ramifications of $\lambda$-calculus that relax
resource-based conditions in different ways.  Currently, we are
studying analogous ramifications of linear $\lambda$-calculus in the
$\mathcal{V}$-equational setting, particularly affine and Cartesian
versions. We are also studying the possibility of adding an
exponential modality in order to obtain a \emph{mixed
  linear-non-linear} calculus~\cite{benton1994mixed}. We also started
to explore different definitions of a morphism between
$\mathcal{V}\lambda$-theories and respective categories. This is the
basis to establish a categorical equivalence between a (quasi)category
of $\mathcal{V}\lambda$-theories and a (quasi)category of
$\VCatSe$-enriched autonomous categories.

Next, our main examples of $\mathcal{V}\lambda$-theories (see
Section~\ref{sec:examples}) used either the Boolean or the metric
quantale.  We would like to study linear $\mathcal{V}\lambda$-theories
whose underlying quantales are neither the Boolean nor the metric one,
for example the ultrametric quantale which is (tacitly) used to
interpret Nakano's guarded $\lambda$-calculus \cite{birkedal10} and
also to interpret a higher-order language for functional reactive
programming~\cite{krishnaswami11}. Another interesting quantale is the
G\"{o}del one which is a basis for fuzzy logic~\cite{denecke13} and
whose $\mathcal{V}$-equations give rise to what we call fuzzy
inequations.

Finally we plan to further explore the connections between our work
and different results on metric universal
algebra~\cite{mardare16,mardare17,rosicky20} and inequational
universal
algebra~\cite{kurz2017quasivarieties,adamek20,rosicky20}. For example,
an interesting connection is that the monad construction presented
in~\cite{mardare16} crucially relies on quotienting  a
pseudometric space into a metric space -- this is a particular case of
quotienting a $\mathcal{V}$-category into a separated
$\mathcal{V}$-category (which we crucially use in our work).

%%
%% Bibliography
%%

%% Please use bibtex, 

\bibliography{biblio}

\appendix

\section{Linear $\mathcal{V}\lambda$-theories and linear quantales}
\label{sec:simpl}
We briefly study linear $\mathcal{V}\lambda$-theories where
$\mathcal{V}$ is a quantale with a \emph{linear order}. The latter
condition is respected by the Boolean and (ultra)metric quantales
(mentioned in the main text). Recall that metric universal
algebra~\cite{mardare16,mardare17} tacitly uses the metric
quantale.

\begin{theorem}\label{theo:simpl}
  Assume that the underlying order of $\mathcal{V}$ is linear and
  consider a (symmetric) linear
  $\mathcal{V}\lambda$-theory. Substituting the rule below on the left
  by the one below on the right does not change the theory.
  \begin{flalign*}
    \infer{v =_{\vee q_i} w}{\forall i \leq n.\ v =_{q_i} w} \hspace{2cm}
    \infer{v =_\bot w}{}
  \end{flalign*}
\end{theorem}

\begin{proof}
  Clearly, the rule on the left subsumes the one on the right by
  choosing $n = 0$. So we only need to show the inverse direction
  under the assumption that $\mathcal{V}$ is linear. Thus, assume that
  $\forall i \leq n.\ v =_{q_i} w$. We proceed by case distinction. If
  $n = 0$ then we need to show that $v =_\bot w$ which is given
  already by the rule on the right.  Suppose now that $n > 0$. Then
  since the order of $\mathcal{V}$ is linear the value $\vee q_i$ must
  already be one of the values $q_i$ and $v =_{q_i} w$ is already part
  of the theory. In other words, in case of $n > 0$ the rule on the
  left is redundant.
\end{proof}

The above result is in accordance with metric universal
algebra~\cite{mardare16,mardare17} which also does not include rule
\textbf{(join)}. Interestingly, however, we still have $v =_\bot w$
for all $\lambda$-terms $v$ and $w$ and such a rule is not present
in~\cite{mardare16,mardare17}. This is explained by the fact that
metric equations in~\cite{mardare16,mardare17} are labelled
\emph{only} by non-negative rational numbers whilst we also permit
infinity to be a label (in our case, labels are given by a basis $B$
which for the metric case corresponds to the \emph{extended}
non-negative rational numbers). All remaining rules of our
$\mathcal{V}$-equational system instantiated to the metric case find a
counterpart in the metric equational system presented
in~\cite{mardare16,mardare17}.

Next, note that if the quantale $\mathcal{V}$ is finite then
for all $q \in \mathcal{V}$ we have $q \ll q$ which means that rule
\textbf{(arch)} is no longer necessary. This observation is applicable
to the Boolean quantale.

\section{Lemmata and omitted proofs}

\begin{proof}[Proof sketch of Lemma~\ref{lem:unique}]
  The proof follows by induction over the structure of
  $\lambda$-terms. Here we only consider the case $f(v_1,\dots,v_n)$,
  because the other cases follow analogously.

  Suppose that $E \vljud f(v_1,\dots,v_n) : \typeA$.  Then according
  to the typing system it is necessarily the case that the previous
  derivations were $\Gamma_i \vljud v_i : \typeA$ for all $i \leq n$
  with $E \in \Shuff(\Gamma_1;\dots;\Gamma_n)$ for some family of
  contexts $(\Gamma_i)_{i \leq n}$. The only room for choice is
  therefore in choosing the contexts $\Gamma_i$. We will show that
  even this choice is unique. Consider two families
  $(\Gamma_i)_{i \leq n}$ and $(\Gamma'_i)_{i \leq n}$ such that
  $\Gamma_i \vljud v_i : \typeA$ and $\Gamma'_i \vljud v_i : \typeA$
  for all $i \leq n$, and moreover
  $E \in \Shuff(\Gamma_1;\dots;\Gamma_n)$ and
  $E \in \Shuff(\Gamma'_1;\dots;\Gamma'_n)$. Since
  $\Gamma_i \vljud v_i : \typeA_i$ and
  $\Gamma'_i \vljud v_i : \typeA_i$ we deduce (by linearity) that
  $\Gamma_i$ is a permutation of $\Gamma'_i$. Consequently, since
  $E \in \Shuff(\Gamma_1,\dots,\Gamma_n)$,
  $E \in \Shuff(\Gamma'_1,\dots,\Gamma'_n)$ and $E$ (by the definition
  of a shuffle) cannot change the relative order of the elements in
  $\Gamma_i$ and $\Gamma'_i$ for all $i \leq n$, it must be the case
  that $\Gamma_i = \Gamma'_i$ for all $i \leq n$. In other words, the
  choice of $(\Gamma_i)_{i \leq n}$ is fixed \emph{a priori}. The
  proof now follows by applying the induction hypothesis to each
  $v_i$.
\end{proof}

\begin{proof}[Proof sketch of Lemma~\ref{lem:exch_subst}]
  We focus first on the exchange rule. The proof follows by induction
  over the structure of derivations. Here we only consider the case
  $\Gamma, x : \typeA, y : \typeB, \Delta \vljud f(v_1,\dots,v_n) :
  \typeC$,  the other cases follow analogously.

  Suppose that
  $\Gamma, x : \typeA, y : \typeB, \Delta \vljud f(v_1,\dots,v_n) :
  \typeC$ with
  $\Gamma, x : \typeA, y : \typeB, \Delta \in
  \Shuff(\Gamma_1;\dots;\Gamma_n)$. We proceed by case distinction:
  assume first that both $x : \typeA$ and $y : \typeB$ are in some
  $\Gamma_i$, with $i \leq n$. We can thus decompose $\Gamma_i$ into
  $\Gamma^1_i, x : \typeA, y : \typeB, \Gamma^2_i$. Then we apply the
  induction hypothesis on $\Gamma_i \vljud v_i : \typeA_i$ and proceed
  by observing that if
  $\Gamma, x : \typeA, y : \typeB, \Delta \in
  \Shuff(\Gamma_1;\dots;(\Gamma^1_i, x : \typeA, y : \typeB,
  \Gamma^2_i); \dots; \Gamma_n)$ then it is also the case that
  $\Gamma, y : \typeB, x : \typeA, \Delta \in
  \Shuff(\Gamma_1;\dots;(\Gamma^1_i, y : \typeB, x : \typeA,
  \Gamma^2_i); \dots; \Gamma_n)$.  Assume now that $x : \typeA$ is in
  some $\Gamma_i$ and $y : \typeB$ is in some $\Gamma_j$ with
  $i \not = j$. Then since
  $\Gamma, x : \typeA, y : \typeB, \Delta \in
  \Shuff(\Gamma_1;\dots;\Gamma_n)$ it must be the case that
  $\Gamma, y : \typeB, x : \typeA, \Delta \in
  \Shuff(\Gamma_1;\dots;\Gamma_n)$ so we only need to apply  rule
  $\rulename{ax}$.

  Let us now focus on the substitution rule. The proof follows by
  induction over the structure of derivations and also by the exchange
  rule that was just proved. We exemplify this with rule
  $\rulename{ax}$. The other cases follow analogously.

  Suppose that $\Gamma, x : \typeA \vljud f(v_1,\dots,v_n) : \typeB$.
  Then for all $i \leq n$ we have $\Gamma_i \vljud v_i : \typeA_i$ and
  $\Gamma, x : \typeA \in \Shuff(\Gamma_1 ; \dots ; \Gamma_n)$.  By
  linearity and by the definition of a shuffle there exists exactly
  one $\Gamma_i$ that can be decomposed into
  $\Gamma_i = \Gamma'_i, x : \typeA$. We then use the induction
  hypothesis to obtain $\Gamma'_i, \Delta \vljud v_i[w/x] : \typeA_i$.
  Now observe that if
  $\Gamma, x : \typeA \in \Shuff(\Gamma_1; \dots ; (\Gamma'_i, x :
  \typeA ); \dots ; \Gamma_n)$ then
  $\Gamma, \Delta \in \Shuff(\Gamma_1; \dots ; (\Gamma'_i, \Delta) ;
  \dots ; \Gamma_n)$. We use this last observation to build
  $\Gamma, \Delta \vljud f(v_1,\dots,v_i[w/x], \dots, v_n) =
  f(v_1,\dots,v_n)[w/x] : \typeB$. % Suppose now that
  % $\Gamma, x : \typeA \vljud v\ \prog{to}\ y.\ u : \typeB$. Then
  % according to the typing system we have
  % $\Gamma_1 \vljud v : \typeC$,
  % $\Gamma_2, y : \typeC \vljud u : \typeB$ and
  % $\Gamma, x : \typeA \in \Shuff(\Gamma_1 ; \Gamma_2)$. We now proceed
  % by case distinction: if $\Gamma_1 = \Gamma'_1, x : \typeA$ then we
  % simply apply the induction hypothesis and the proof follows
  % straightforwardly. If $\Gamma_2 = \Gamma'_2, x : \typeA$ then we
  % obtain $\Gamma'_2, x : \typeA, y : \typeC \vljud u : \typeB$, which
  % by the exchange rule entails
  % $\Gamma'_2, y : \typeC, x : \typeA \vljud u : \typeB$. We now apply
  % the induction hypothesis to obtain
  % $\Gamma'_2, y : \typeC, \Delta \vljud u[w/x] : \typeB$ and applying
  % the exchange rule again gives
  % $\Gamma'_2, \Delta, y : \typeC \vljud u[w/x] : \typeB$. Now observe
  % that if
  % $\Gamma, x : \typeA \in \Shuff(\Gamma_1; \Gamma'_2, x : \typeA)$
  % then $\Gamma, \Delta \in \Shuff(\Gamma_1; \Gamma'_2, \Delta)$. We use
  % this last observation and the rule $\rulename{seq}$ to build
  % $\Gamma, \Delta \vljud v\ \prog{to}\ y.\ (u[w/x]) = (v\ \prog{to}\ y.\
  % u)[w/x] : \typeB$.
\end{proof}
In order to keep calculations in the following proofs legible we will
sometimes abbreviate a denotation $\sem{\Gamma \vljud v : \typeA}$ to
$\sem{\Gamma \vljud v}$ or even just $\sem{v}$.

\begin{lemma}[Exchange and Substitution]
  \label{lem:exch_subst_inter}
  Consider judgements
  $\Gamma, x : \typeA, y : \typeB, \Delta \vljud v : \typeC$,
  $\Gamma, x : \typeA \vljud v : \typeB$, and
  $\Delta \vljud w : \typeA$. Then the following equations hold in
  every autonomous category $\catC$:
  \begin{flalign*}
    \sem{\Gamma, x : \typeA, y : \typeB, \Delta \vljud v : \typeC} & =
    \sem{\Gamma, y : \typeB, x : \typeA, \Delta \vljud v : \typeC}
    \comp \exch_{\Gamma, \underline{\typeA,\typeB}, \Delta} \\
    \sem{\Gamma, \Delta \vljud v[w / x] : \typeB} & =
    \sem{\Gamma, x : \typeA \vljud v : \typeB} \comp
    \join_{\Gamma;\typeA} \comp (\id \otimes \sem{\Delta \vljud w : \typeA})
    \comp \spl_{\Gamma;\Delta}
  \end{flalign*}
\end{lemma}

\begin{proof}[Proof sketch]
  For both cases the proof follows by induction over the structure of
  derivations. Here we only consider rule $\rulename{ax}$, because the
  other ones follow analogously. In many of the calculations below we
  will tacitly perform simple diagram chases that take advantage of
  naturality, functoriality, and the coherence theorem of symmetric
  monoidal categories.

  We start with the exchange property. Suppose that
  $\Gamma, x : \typeA, y : \typeB, \Delta \vljud f(v_1,\dots,v_n) :
  \typeC$.  We proceed by case distinction: first, consider the case
  in which $x : \typeA \in \Gamma_i$ and $y : \typeB \in \Gamma_j$
  with $i \not = j$. The proof then follows directly by observing that
  the two corresponding shuffling morphisms
  $\sh_{\Gamma, \typeA, \typeB, \Delta} : \sem{\Gamma, x : \typeA, y :
    \typeB, \Delta} \to \sem{\Gamma_1,\dots,\Gamma_n}$ and
  $\sh_{\Gamma, \typeB, \typeA, \Delta} : \sem{\Gamma, y : \typeB, x :
    \typeA, \Delta} \to \sem{\Gamma_1,\dots,\Gamma_n}$ satisfy the
  equation
  $\sh_{\Gamma,\typeB,\typeA,\Delta} \comp
  \exch_{\Gamma,\underline{\typeA,\typeB},\Delta} =
  \sh_{\Gamma,\typeA,\typeB,\Delta}$. Consider now the case in which
  $x : \typeA \in \Gamma_i$ and $y : \typeB \in \Gamma_i$ for some
  $i \leq n$. We then calculate:
  \begin{flalign*}
   & \,
   \sem{\Gamma, x : \typeA, y : \typeB, \Delta \vljud f(v_1,\dots,v_n) :
   \typeC}
   = \sem{f} \comp (\sem{v_1} \otimes \dots \otimes \sem{v_i}
   \otimes \dots \otimes
   \sem{v_n}) \comp
   \spl_{\Gamma_1;\dots;\Gamma_n} \comp \sh_{\Gamma, \typeA, \typeB, \Delta} \\
   & =  \sem{f} \comp (\sem{v_1}
   \otimes \dots \otimes (\sem{v_i} \comp
   \exch_{\Gamma^1_i,\underline{\typeA,\typeB},\Gamma^2_i})
   \otimes \dots \otimes
   \sem{v_n}) \comp \spl_{\Gamma_1;\dots;\Gamma_n} \comp
   \sh_{\Gamma, \typeA, \typeB, \Delta} \\
   & =  \sem{f} \comp (\sem{v_1}
   \otimes \dots \otimes \sem{v_n}) \comp (\id \otimes \dots \otimes
   \exch_{\Gamma^1_i,\underline{\typeA,\typeB},\Gamma^2_i} \otimes \dots \otimes \id)
   \comp \spl_{\Gamma_1;\dots;\Gamma_n} \comp
   \sh_{\Gamma, \typeA, \typeB, \Delta} \\
   & =  \sem{f} \comp (\sem{v_1}
   \otimes \dots \otimes \sem{v_n}) 
   \comp \spl_{\Gamma'_1;\dots;\Gamma'_n} \comp
   \exch_{\Gamma_1,\dots,\Gamma^1_i, \underline{\typeA, \typeB}, \Gamma^2_i, \dots
   \Gamma_n} \comp
   \sh_{\Gamma, \typeA, \typeB, \Delta} \\
   & =  \sem{f} \comp (\sem{v_1}
   \otimes \dots \otimes \sem{v_n}) 
   \comp \spl_{\Gamma'_1;\dots;\Gamma'_n} \comp \sh_{\Gamma, \typeB, \typeA, \Delta}
   \comp \exch_{\Gamma, \underline{\typeA, \typeB}, \Delta} \\
   & = \sem{\Gamma, y : \typeB, x : \typeA, \Delta \vljud
   f(v_1,\dots,v_n) : \typeC} \comp \exch_{\Gamma, \underline{\typeA, \typeB}, \Delta}
   \end{flalign*}

 Let us now focus on proving the substitution lemma for
 rule $\rulename{ax}$:
 \begin{flalign*}
   & \, \sem{\Gamma, \Delta \vljud f(v_1,\dots,v_n) [w/x] : \typeB} =
   \sem{\Gamma, \Delta \vljud f(v_1,\dots,v_i[w/x], \dots, v_n)} \\
   & = \sem{f} \comp (\sem{v_1}
   \otimes \dots \otimes \sem{v_i[w/x]} \otimes \dots \otimes
   \sem{v_n}) \comp \spl_{\Gamma_1;\dots;\Gamma_n} \comp
   \sh_{\Gamma, \Delta}  \\
   & = \sem{f} \comp (\sem{v_1}
    \otimes \dots \otimes (\sem{v_i} \comp \join_{\Gamma'_i;\typeA} \comp
    (\id \otimes \sem{w}) \comp \spl_{\Gamma'_i;\Delta}) \otimes \dots \otimes
   \sem{v_n}) \comp \spl_{\Gamma_1;\dots;\Gamma_n} \comp
   \sh_{\Gamma, \Delta}  \\
   & = \sem{f} \comp (\sem{v_1} \otimes \dots \otimes \sem{v_n}) \comp 
   (\id \otimes \dots \otimes (\join_{\Gamma'_i;\typeA} \comp
   (\id \otimes \sem{w}) \comp
   \spl_{\Gamma'_i;\Delta}) \otimes \dots \otimes \id)
   \comp \spl_{\Gamma_1;\dots;\Gamma_n} \comp \sh_{\Gamma, \Delta}  \\
   & = \sem{f} \comp (\sem{v_1} \otimes \dots  \otimes \sem{v_n}) 
   \comp \spl_{\Gamma_1;\dots;\Gamma'_i,\typeA;\dots\Gamma_n}
   \comp \sh_{\Gamma, \typeA} \comp \join_{\Gamma;\typeA} \comp
   (\id \otimes \sem{w}) \comp \spl_{\Gamma;\Delta} \\
   & = \sem{\Gamma, x : \typeA \vljud f(v_1,\dots,v_n)} \comp
   \join_{\Gamma;\typeA} \comp
   (\id \otimes \sem{w}) \comp \spl_{\Gamma;\Delta}
   \end{flalign*}
\end{proof}

\begin{proof}[Proof sketch of Theorem~\ref{theo:bsound}]
  The proof follows by an appeal to Lemma~\ref{lem:exch_subst_inter},
  the coherence theorem for symmetric monoidal categories, and
  naturality. We exemplify this with the commuting conversions.
    \begin{flalign*}
    & \, \sem{\Gamma, \Delta,E \vljud u[v\ \prog{to}\ \ast.\ w/x]
      : \typeB} = \sem{u}  \comp
    \join_{\Gamma;\typeA} \comp (\id \otimes \sem{v\ \prog{to}\ \ast.\ w})
    \comp \spl_{\Gamma; \Delta,E} \\
    & = \sem{u} \comp \join_{\Gamma;\typeA} \comp (\id \otimes
    (\sem{w} \comp \lambda \comp (\sem{v} \otimes \id) \comp
    \spl_{\Delta;E})) \comp \spl_{\Gamma;\Delta,E} \\
    & = \sem{u} \comp \join_{\Gamma;\typeA} \comp (\id \otimes \sem{w})
    \comp (\id \otimes
    (\lambda \comp (\sem{v} \otimes \id) \comp
    \spl_{\Delta;E})) \comp \spl_{\Gamma;\Delta,E} \\
    & = \sem{u} \comp \join_{\Gamma;\typeA} \comp (\id \otimes \sem{w})
    \comp \spl_{\Gamma;E} \comp \join_{\Gamma;E} \comp (\id \otimes
    (\lambda \comp (\sem{v} \otimes \id) \comp
    \spl_{\Delta;E})) \comp \spl_{\Gamma;\Delta,E} \\
    & = \sem{u[w/x]} \comp \join_{\Gamma;E} \comp (\id \otimes
    (\lambda \comp (\sem{v} \otimes \id) \comp
    \spl_{\Delta;E})) \comp \spl_{\Gamma;\Delta,E} \\
    & = \sem{u[w/x]} \comp \lambda \comp (\sem{v} \otimes \id) \comp
    \spl_{\Delta;\Gamma,E} \comp \sh_{\Gamma,\Delta,E} \\
    & = \sem{\Gamma,\Delta,E \vljud v\ \prog{to}\ \ast.\ u[w/x] :
      \typeB}
  \end{flalign*}
  The one but last step amounts to a diagram chase that recurs to
  naturality and the coherence theorem of symmetric monoidal
  categories.
  \begin{flalign*}
    & \, \sem{\Gamma,\Delta,E \vljud u[\prog{pm}\ v\ \prog{to}\ x \otimes y.\
      w /z] : \typeB} = \sem{\Gamma, z : \typeA \vljud u}
    \comp \join_{\Gamma;\typeA}
    \comp (\id \otimes \sem{\prog{pm}\ v\ \prog{to}\ x \otimes y.\ w})
    \comp \spl_{\Gamma;\Delta,E} \\
    & = \sem{u} \comp \join_{\Gamma;\typeA} \comp
    (\id \otimes (\sem{w} \comp \join_{E;\typeC;\typeD} \comp \alpha \comp \sw
    \comp (\sem{v} \otimes \id) \comp \spl_{\Delta;E})) \comp \spl_{\Gamma;\Delta,E}
    \\
    & = \sem{u} \comp \join_{\Gamma;\typeA} \comp (\id \otimes \sem{w}) \comp
    (\id \otimes (\join_{E;\typeC;\typeD} \comp \alpha \comp \sw
    \comp (\sem{v} \otimes \id) \comp \spl_{\Delta;E})) \comp \spl_{\Gamma;\Delta,E}
    \\
    & = \sem{u} \comp \join_{\Gamma;\typeA} \comp (\id \otimes \sem{w}) 
    \comp \spl_{\Gamma;E,\typeC,\typeD} \comp \join_{\Gamma;E,\typeC,\typeD} \comp
    (\id \otimes (\join_{E;\typeC;\typeD} \comp \alpha \comp \sw
    \comp (\sem{v} \otimes \id) \comp \spl_{\Delta;E})) \comp \spl_{\Gamma;\Delta,E}
    \\
    & = \sem{u[w/z]} \comp \join_{\Gamma;E,\typeC,\typeD} \comp
    (\id \otimes (\join_{E;\typeC;\typeD} \comp \alpha \comp \sw
    \comp (\sem{v} \otimes \id) \comp \spl_{\Delta;E})) \comp \spl_{\Gamma;\Delta,E}
    \\
    % & = \sem{u[w/z]} \comp \join_{\Gamma;E,\typeC,\typeD} \comp
    % (\id \otimes (\join_{E;\typeC;\typeD} \comp \alpha)) \comp (\id \otimes (\sw
    % \comp (\sem{v} \otimes \id) \comp \spl_{\Delta;E})) \comp \spl_{\Gamma;\Delta,E}
    % \\
    % & = \sem{u[w/z]} \comp \join_{\Gamma,E;\typeC;\typeD} \comp \alpha
    % \comp (\join_{\Gamma;E} \otimes \id) \comp \alpha \comp (\id \otimes (\sw
    % \comp (\sem{v} \otimes \id) \comp \spl_{\Delta;E})) \comp \spl_{\Gamma;\Delta,E}
    % \\
    % & = \sem{u[w/z]} \comp \join_{\Gamma,E;\typeC;\typeD} \comp \alpha
    % \comp (\join_{\Gamma;E} \otimes \id) \comp \alpha \comp
    % (\id \otimes \sw) \comp
    % (\id \otimes (\sw
    % \comp (\sem{v} \otimes \id) \comp \spl_{\Delta;E})) \comp \spl_{\Gamma;\Delta,E}
    % \\
    % & = \sem{u[w/z]} \comp \join_{\Gamma,E;\typeC;\typeD} \comp \alpha
    % \comp \sw \comp (\id \otimes \join_{\Gamma;E}) \comp (\id \otimes \sw)
    % \comp \alpha^{-1} \comp \sw \comp
    % (\id \otimes ((\sem{v} \otimes \id) \comp \spl_{\Delta;E})) \comp
    % \spl_{\Gamma;\Delta,E} \\
    % & = \sem{u[w/z]} \comp \join_{\Gamma,E;\typeC;\typeD} \comp \alpha
    % \comp \sw \comp (\id \otimes \join_{\Gamma;E}) \comp (\sem{v} \otimes \id)
    % \comp (\id \otimes \sw) \comp \alpha^{-1} \comp \sw \comp
    % (\id \otimes \spl_{\Delta;E}) \comp \spl_{\Gamma;\Delta,E} \\
    % & = \sem{u[w/z]} \comp \join_{\Gamma,E;\typeC;\typeD} \comp \alpha
    % \comp \sw \comp (\id \otimes \join_{\Gamma;E}) \comp (\sem{v} \otimes \id)
    % \comp (\id \otimes \spl_{\Gamma;E}) \comp \spl_{\Delta;\Gamma,E} \comp
    % \sh_{\Gamma,\Delta,E} \\
    & = \sem{u[w/z]} \comp \join_{\Gamma,E;\typeC;\typeD} \comp \alpha
    \comp \sw \comp (\sem{v} \otimes \id)
    \comp \spl_{\Delta;\Gamma,E} \comp
    \sh_{\Gamma,\Delta,E} \\
    & = \sem{\Gamma,\Delta,E \vljud \prog{pm}\ v\ \prog{to}\ x \otimes y.\
      u [w/z] : \typeB}
  \end{flalign*}
  The one but last step amounts to a diagram chase that recurs to
  naturality and the coherence theorem of symmetric monoidal
  categories.
\end{proof}

\begin{proof}[Proof of Theorem~\ref{theo:auto}]
  The proof follows by showing that the closed monoidal structure of
  $\VCat$ preserves symmetry and separation. It is immediate for
  symmetry. For separation, note that since $\mathcal{V}$ is integral
  the inequation $x \otimes y \leq x$ holds for all
  $x,y \in \mathcal{V}$. It follows that the monoidal structure
  preserves separation. The fact that the closed structure also
  preserves separation uses the implication
  $x \leq \bigwedge A \Rightarrow \forall a \in A.\ x \leq a$ for all
  $x \in X, A \subseteq X$.
\end{proof}

\begin{proof}[Proof of Theorem~\ref{theo:vari}]
  Let us denote by $\Lang^\lambda(\catC)$ the \emph{linear
    $\lambda$-theory} generated from $\catC$. According to
  Theorem~\ref{theo:internalb}, the category
  $\Syn(\Lang^\lambda(\catC))$ (\ie the syntactic category generated
  from $\Lang^\lambda(\catC)$) is locally small whenever $\catC$ is
  locally small.  Then consider two types $\typeA$ and $\typeB$.  We
  will prove our claim by taking advantage of the axiom of replacement
  in ZF set-theory, specifically by presenting a \emph{surjective}
  map,
  \begin{flalign*}
    \Syn(\Lang^\lambda(\catC))(\typeA,\typeB) \longrightarrow
    \Syn(\Lang(\catC))(\typeA,\typeB)
  \end{flalign*}
  The crucial observation is that if $v = w$ in $\Lang^\lambda(\catC)$
  then $v =_\top w$ and $w =_\top v$ in $\Lang(\catC)$. This is
  obtained by the definition of a model, the definition of a
  $\mathcal{V}$-category, and the definition of $\Lang(\catC)$. This
  observation allows to establish the surjective map that sends $[v]$
  to $[v]$, \ie it sends the equivalence class of $v$ as a
  $\lambda$-term in $\Lang^\lambda(\catC)$ into the equivalence class
  of $v$ as a $\lambda$-term in $\Lang(\catC)$.
\end{proof}

\begin{proof}[Proof of Theorem~\ref{theo:Ban}]
  The autonomous structure of $\Ban$ is well-known~\cite{DK20a,K81c},
  so let us focus on  showing that it is a \emph{$\Met$-enriched}
  autonomous category. The enrichment is simply given by distance
  function induced by the operator norm, thus if $S,T\in\Ban(X,Y)$,
\[
d(S,T)=\bigvee\{\norm{(S-T)(x)} \mid \norm{x}\leq 1\}
\]

\vspace{0.5em}
\noindent \textbf{Composition is a short map.}

\noindent Let $T,T'\in \Ban(X,Y)$ and $S,S'\in \Ban(Y,Z)$, we compute:
\begin{align*}
d(ST,S'T')&\defeq \bigvee\{\norm{ST(x)-S'T'(x)} \mid \norm{x}\leq 1\} \\
&=\bigvee\{\norm{ST(x)-ST'(x)+ST'(x)-S'T'(x)}\mid \norm{x}\leq 1\}\\
&\leq\bigvee\{\norm{ST(x)-ST'(x)}+\norm{ST'(x)-S'T'(x)} \mid \norm{x}\leq 1\}\\
&=\bigvee\{\norm{ST(x)-ST'(x)}\mid \norm{x}\leq 1\}+\bigvee\{\norm{ST'(x)-S'T'(x)} \mid \norm{x}\leq 1\}\\
&=\bigvee\{\norm{S(T-T')(x)}\mid \norm{x}\leq 1\}+\bigvee\{\norm{ST'(x)-S'T'(x)}\mid \norm{x}\leq 1\}\\
&\stackrel{(\star)}{\leq} \bigvee\{\norm{T(x)-T'(x)}\mid \norm{x}\leq 1\}+\bigvee\{\norm{ST'(x)-S'T'(x)}\mid \norm{x}\leq 1\}\\
&\leq \bigvee\{\norm{T(x)-T'(x)}\mid \norm{x}\leq 1\}+\bigvee\{\norm{S(y)-S'(y)}\mid\norm{y}\leq 1\}\\
&\defeq d(T,T') + d(S,S')\\
&\defeq d((T,S),(T',S'))
\end{align*}
where $(\star)$ follows from the fact that,
\[
\norm{S(T-T')(x)}=\norm{T(x)-T'(x)}\norm{S\left(\frac{(T-T')(x)}{\norm{T(x)-T'(x)}}\right)}\leq \norm{T(x)-T'(x)}
\]
by linearity of $S$ and by the fact that $\norm{S}\leq 1$. This shows
that $\Ban$ is $\Met$-enriched. We now turn to the first clause of
Definition \ref{defn:enr_aut}.

\vspace{1em}
\noindent \textbf{The monoidal operation is an enriched bi-functor.}

\noindent Note first that if $S\in\Ban(X,Y)$ and $T,T'\in \Ban(X',Y')$
then,
\[
S\ptp T-S\ptp T'=S\ptp(T-T')
\]
Indeed, since $S\ptp T$ is the unique linear operator such that
$S\ptp T(x\otimes x') = S(x)\otimes T(x')$, we can reason pointwise
and get,
\begin{align*}
(S\ptp T-S\ptp T')(x\otimes x')  & = (S\ptp T)(x\otimes x') - (S\ptp T')(x\otimes x')\\
&=S(x)\otimes T(x')  - S(x)\otimes T'(x')\\
&=S(x)\otimes (T(x')-T'(x'))\\
&=(S\ptp (T-T'))(x\otimes x')
\end{align*}
where the penultimate step follows from the basic definition of the
tensor product of vector spaces. Now we can show that for any Banach
spaces $X,X',Y,Y'$ the projective tensor map,
\[
\Ban(X,Y)\otimes\Ban(X',Y')\to\Ban(X\ptp X', Y\ptp Y'),  (S,T)\mapsto S\ptp T
\]
where $\otimes$ once again denotes the monoidal operation in $\Met$,
is short. We simply compute,
\begin{align*}
d(S\ptp T, S'\ptp T')&\defeq \norm{S\ptp T - S'\ptp T'}\\
&= \norm{S\ptp T -S\ptp T' + S\ptp T'- S'\ptp T'}\\
&\leq  \norm{S\ptp T -S\ptp T'} +\norm{ S\ptp T'- S'\ptp T'}\\
&=\norm{S\ptp(T-T')} + \norm{(S-S')\ptp T'}\\
&=\norm{S}\norm{T-T'} + \norm{S-S'}\norm{T'}\\
&\leq \norm{T-T'} + \norm{S-S'} \defeq d((S,T),(S,T'))
\end{align*}
where the last step uses the fact that $\norm{S},\norm{T}\leq 1$ and the penultimate step uses the basic fact that $\norm{S\ptp T}=\norm{S}\norm{T}$ (see \cite[\S 2.1]{ryan2013introduction}). Finally, we show the second clause of Definition \ref{defn:enr_aut}.

\vspace{1em}
\noindent
\textbf{The adjunction $-\ptp Y \vdash Y\multimap -$ is a $\Met$-adjunction.}

\noindent
The fact that the maps,
\[
\Ban(X,X')\to \Ban(X\ptp Y, X\ptp Y),\ f\mapsto f\ptp \id_Y
\]
are short follows by re-writing them as,
\[
  \Ban(X,X')\simeq \Ban(X,X')\otimes
  1\longrightarrow\Ban(X,X')\otimes\Ban(Y,Y)\longrightarrow\Ban(X\ptp
  X', Y\ptp Y)
\]
and the fact that the monoidal operation of $\Ban$ is an enriched
bi-functor. 
Similarly, the map,
\[
\Ban(X,X')\to \Ban(Y\multimap X, Y\multimap X'),\ S \mapsto (T\mapsto S T)
\]
is short. This is a consequence of the following fact. Consider two
operators $S,S' \in \Ban(X,X')$. For all bounded operators
$T \in Y \multimap X$ with $\norm{T} \leq 1$ we have
$d(S,S') \defeq \norm{S - S'} \geq \norm{ST - S'T}$, which provides,
\[
  d(S,S') \geq \bigvee \{ \norm{ ST - S'T } \mid \norm{T} \leq 1 \} \defeq
  d(T \mapsto S  T, T \mapsto S'T)
\]
Finally, we need to show that the adjunction
$-\ptp Y \vdash Y \multimap -$ defines an isometry,
\[
\Ban(X\ptp Y, Z)\simeq \Ban(X,Y\multimap Z).
\]
Indeed,  the bijection from left to right is defined by,
\[
T\mapsto (Y\multimap - )(T)\comp \eta_Y
\]
where $\eta$ is the unit of the adjunction. Since $Y \multimap -$ is
$\Met$-enriched, the assignment $T\mapsto(Y\multimap -)(T)$ is short,
and composition by $\eta_Y$ is short.  Thus the invertible map
$\Ban(X\ptp Y, Z)\to \Ban(X,Y\multimap Z)$ is short. By a similar
argument using the co-unit of the adjunction and the fact that
$-\ptp Y$ is $\Met$-enriched we get that the invertible map
$\Ban(X,Y\multimap Z) \to \Ban(X\ptp Y, Z)$ is also short. It follows
that both maps must be invertible isometries.
\end{proof}

\noindent
\textbf{Proof that the axiom \eqref{eq:bern} is sound.}

\noindent The total variation distance between $\sem{\bern(x,y,p)}$ and $\sem{\bern(x,y,q)}$ with $p,q\in[0,1]$ is given by:
\begin{flalign*}
&\bigvee_A \abs{\int\hspace{-4pt}\int \hspace{-1pt} p\delta_u(A)\hspace{-2pt} +\hspace{-2pt} (1\hspace{-2pt}-\hspace{-2pt}p)\delta_v(A)~d\sem{x}(du)d\sem{y}(dv)\hspace{-2pt} - \hspace{-3pt}\int\hspace{-4pt}\int \hspace{-1pt} q\delta_u(A) \hspace{-2pt}+ \hspace{-2pt}(1\hspace{-2pt}-\hspace{-2pt}q)\delta_v(A)~d\sem{x}(du)d\sem{y}(dv)}
\\
=&\bigvee_A\abs{\int\hspace{-4pt}\int (p-q)\delta_u(A) + ((1-p)-(1-q))\delta_v(A)~d\sem{x}(du)d\sem{y}(dv)}
\\
=&\bigvee_A\abs{(p-q)\sem{x}(A) + (q-p)\sem{y}(A)}
\\
=&\bigvee_A\abs{(p-q)(\sem{x}(A) -\sem{y}(A))}
\\
\leq & ~\abs{p-q}
\end{flalign*}

\end{document}